\documentclass[journal,final,comsoc]{IEEEtran}
\pdfoutput=1

\usepackage{graphicx}
\usepackage{cite}

\usepackage{balance}

\usepackage{hyperref}

\usepackage{bm}
\usepackage{amsmath} 

\usepackage{mathrsfs}

\usepackage{multirow}
\usepackage{dcolumn}
\usepackage{threeparttable}

\usepackage{color}

\usepackage{algorithm}
\usepackage{algorithmic}

\ifCLASSOPTIONcompsoc
    \usepackage[caption=false, font=normalsize, labelfont=sf, textfont=sf]{subfig}
\else
\usepackage[caption=false, font=footnotesize]{subfig}
\fi

\usepackage{amsthm}

\renewcommand\IEEEkeywordsname{Keywords}

\title{Improving the Survivability of Clustered Interdependent Networks by Restructuring Dependencies}

\author{Genya Ishigaki, \IEEEmembership{Student Member, IEEE,} Riti Gour, \IEEEmembership{Student Member, IEEE,} and~Jason~P.~Jue,~\IEEEmembership{Senior~Member,~IEEE}
\thanks{Manuscript submitted \today.} 
\thanks{Genya Ishigaki, Riti Gour, and Jason P. Jue are with the Department of Computer Science at The University of Texas at Dallas, Richardson Texas 75080, USA (Email: \{gishigaki, rgour, jjue\}@utdallas.edu).

An earlier version of this paper has been presented at IEEE International Conference on Communications (ICC) 2018.
}
}

\theoremstyle{definition}

\newtheorem{lemma}{Lemma}
\newtheorem*{problem*}{Problem}

\newtheorem{definition}{Definition}
\newtheorem*{remark*}{Remark}

\def\TCOMscale{0.5}

\def\indeg{\mathrm{deg_{in}}}
\def\outdeg{\mathrm{deg_{out}}}
\def\inarcs{A_{\mathrm{in}}}

\newcounter{sec}
\newcounter{index}[sec]
\def\whichyear{2019}
\def\doiID{http://dx.doi.org/10.1109/TCOMM.2018.2889983}
\usepackage{tikz}
\usepackage{textcomp}
\usepackage{lipsum}
\newcommand\copyrighttext{%
  \footnotesize \textcopyright \whichyear \hspace{0.2em} IEEE. Personal use is permitted. This is the author's version of an article that has been published in this journal. Changes were made to this version by the publisher prior to publication. The final version of record is available at \url{\doiID}.}
\newcommand\copyrightnotice{%
\begin{tikzpicture}[remember picture,overlay]
\node[anchor=south,yshift=10pt] at (current page.south) {\fbox{\parbox{\dimexpr\textwidth-\fboxsep-\fboxrule\relax}{\copyrighttext}}};
\end{tikzpicture}%
}
  
\begin{document}
\maketitle

\copyrightnotice

\IEEEpeerreviewmaketitle

\begin{abstract}
The interdependency between different network layers is commonly observed in Cyber Physical Systems and communication networks adopting the dissociation of logic and hardware implementation, such as Software Defined Networking and Network Function Virtualization. This paper formulates an optimization problem to improve the survivability of interdependent networks by restructuring the provisioning relations. A characteristic of the proposed algorithm is that the continuous availability of the entire system is guaranteed during the restructuring of dependencies by the preservation of certain structures in the original networks. Our simulation results demonstrate that the proposed restructuring algorithm can substantially enhance the survivability of interdependent networks, and provide insights into the ideal allocation of dependencies.
\end{abstract}

\section{Introduction}

\IEEEPARstart{M}{any} network systems encompass layering and integration of the layers in both explicit and implicit manners. For example, Software Defined Networking (SDN) decouples the control logic from forwarding functions to realize the flexibility and agility of communication networks. Also, Network Function Virtualization (NFV) involves separation of network function logic from hardware. The concept of separating logic from hardware implementations is also commonly adopted in Cyber Physical Systems (CPS), such as smart grids, in which computing capability manages physical entities.

The dissociation of logic and functions, which is effective for system flexibility, has accelerated the amount of layering and obscure dependencies in network systems. The work \cite{7367911} on software defined optical networks points out the dependency of logical nodes on physical nodes that provide physical paths for connections among logical nodes, as well as the dependency of physical nodes on the logical nodes \textcolor{black}{through SDN control messages, which define the operations} of the physical nodes. Similarly, it is revealed that NFV embraces the interdependency between Virtual Network Functions (VNF) and physical servers hosting the VNFs, when a virtualization orchestrator is recognized as one of the VNFs \cite{7498079}. Furthermore, the integration of a control information network and an electricity network seen in smart grids is a typical example of the interdependency of two different layers in CPSs \cite{infra_interdep_review}. This tendency of layering and collaborative functionality of layered networks is likely to be more evident for next-generation network systems.

\def\scale{0.45}
\begin{figure} 
\centering
	\subfloat[An interdependent network with two constituent graphs representing physical and logical network.\label{ex1}]{\includegraphics[width=\scale\linewidth]{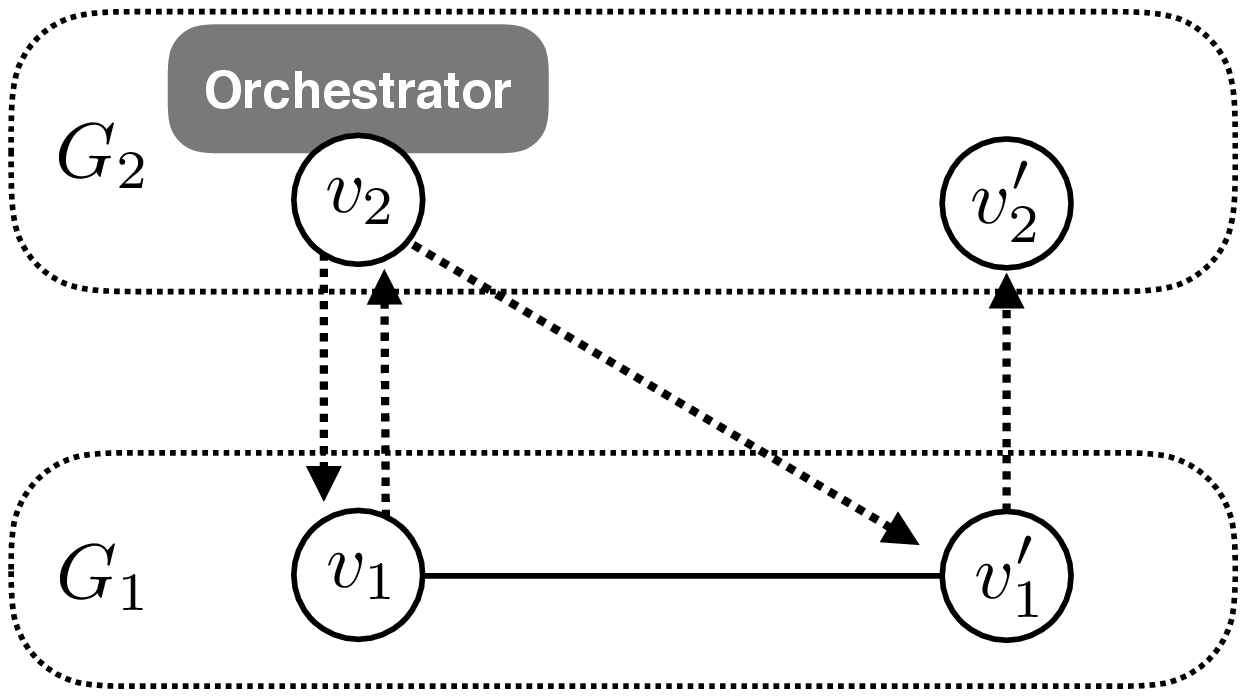} \label{1a}}\hfill
	\subfloat[Initial failure at a physical server $v_1$.\label{ex2}]{\includegraphics[width=\scale\linewidth]{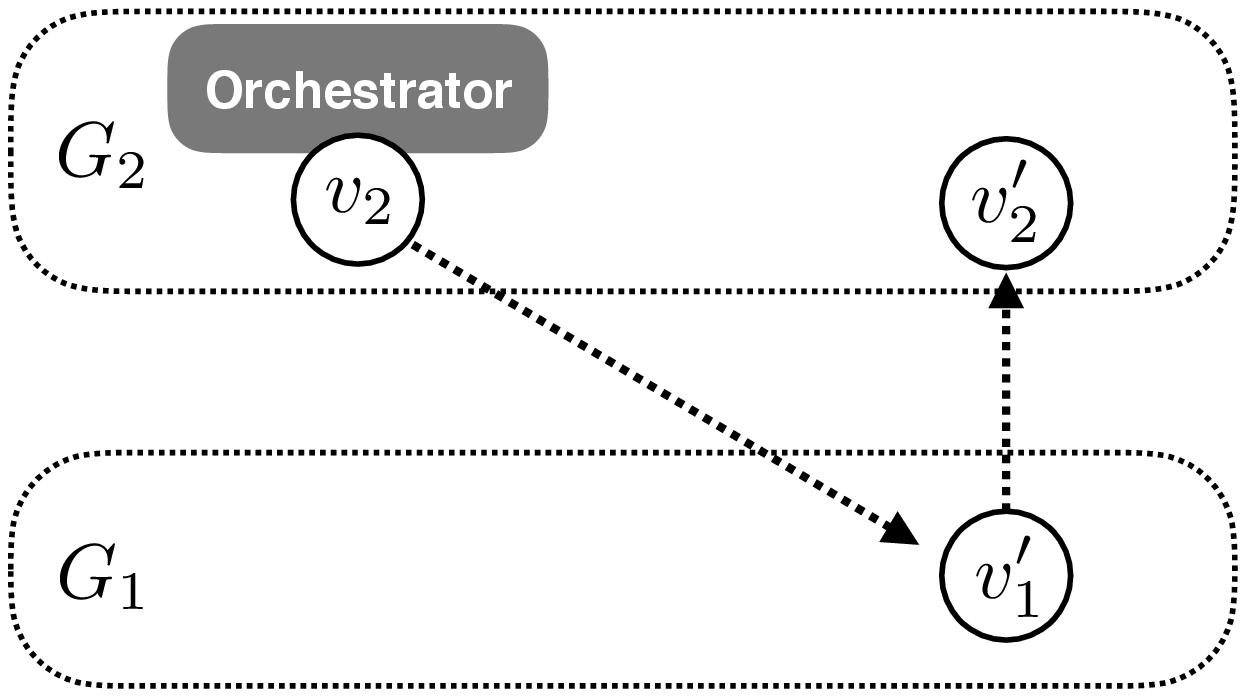} \label{1b}}\hfill\\
	\subfloat[Cascading failure affecting a logical node $v_2$.\label{ex3}]{\includegraphics[width=\scale\linewidth]{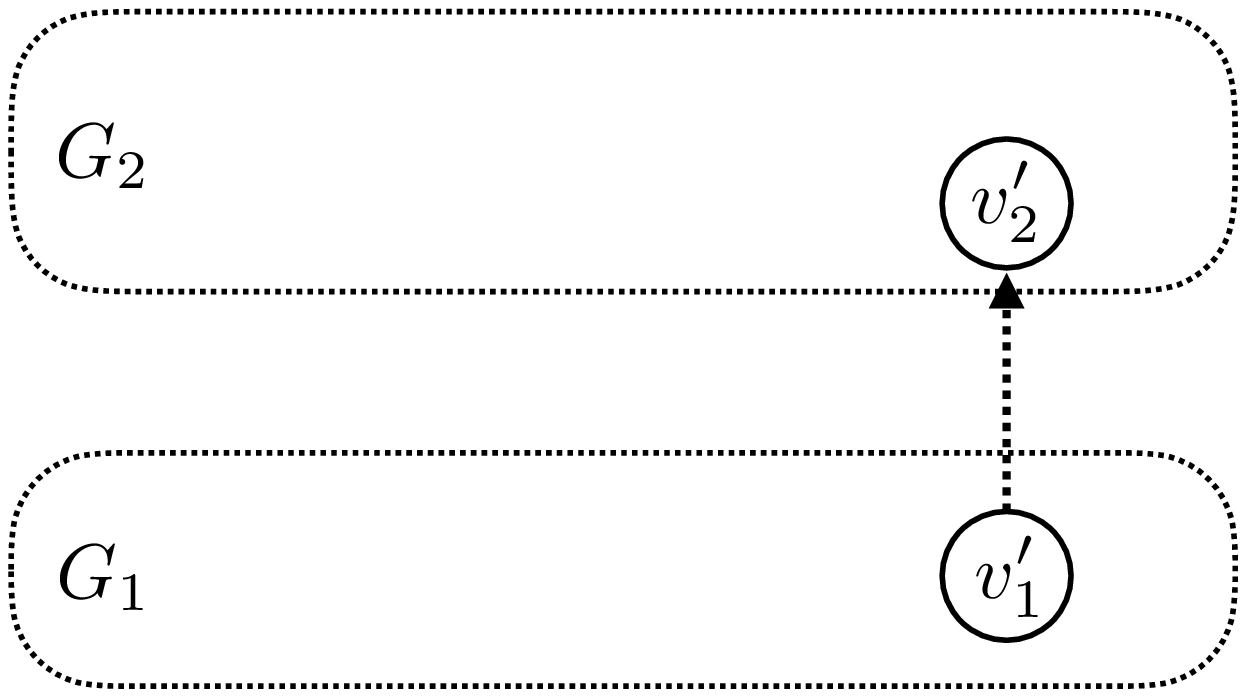} \label{1c}}\hfill
	\subfloat[Cascading failure affecting a physical server $v'_1$. The entire network becomes nonfunctional.\label{ex4}]{\includegraphics[width=\scale\linewidth]{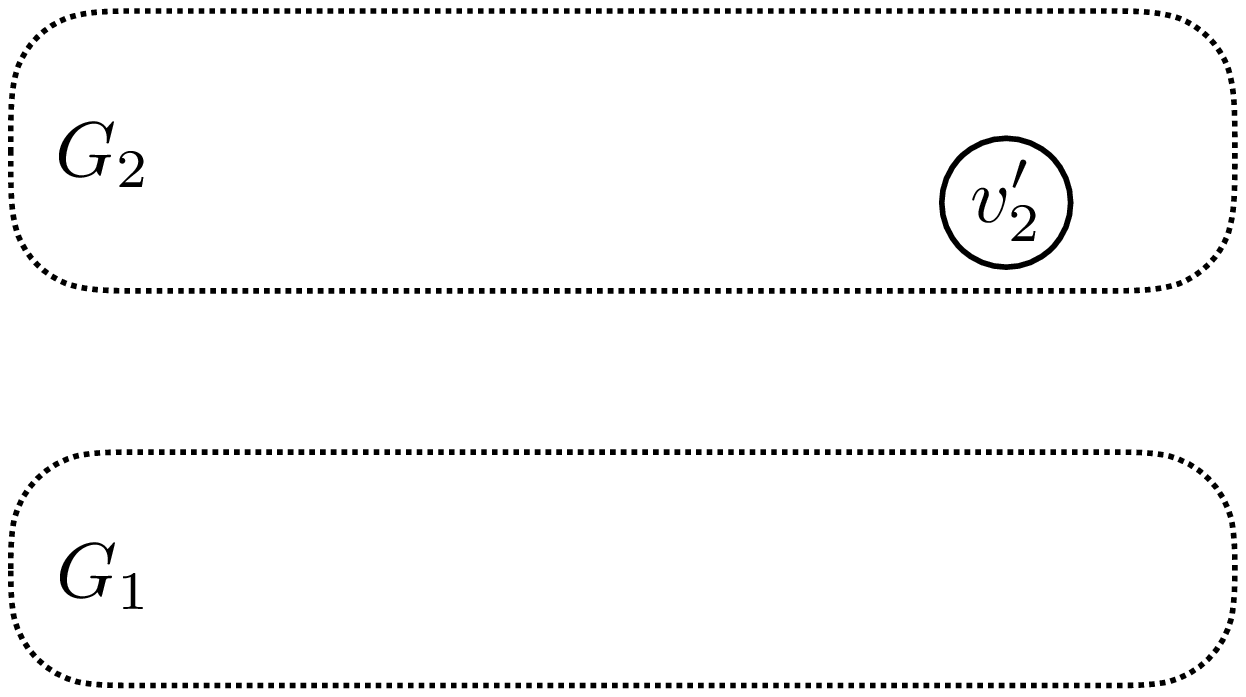} \label{1d}}
\caption{An example of cascading failure in an interdependent network representing the dependency between physical servers and NFVs. \label{fig1}} 
\end{figure}

However, it has been revealed that certain types of dependencies between different layers of networks can deteriorate the robustness of the entire interdependent system \cite{generalsurvey_cascading}. Consecutive multiple failure phenomena called \textit{cascading failures} exemplify the unique fragility of such network systems. In networks without interdependencies, a failure would influence a certain part of a network. Nonetheless, in networks with interdependencies, some nodes that are not directly connected to the failed portion can become nonfunctional due to the loss of service provisioning from nodes in other layers, which are directly influenced by the initial failure. 

Fig. \ref{fig1} shows an example of such a cascading failure, which starts as a single node failure of $v_1$ and results in the entire network failure. Suppose that a network $G_1$ consists of physical servers $v_1$ and $v'_1$, and $G_2$ represents logical computing nodes $v_2$ and $v'_2$ hosting VNFs. The orchestrator, which coordinates the mapping between physical and logical layer, is realized as one of the VNFs on $v_2$. The arcs from $G_1$ to $G_2$ ($(v_1, v_2), (v'_1, v'_2)$) illustrate the dependency of NFVs or computing nodes on the physical servers, while the arcs from $G_2$ to $G_1$ ($(v_2, v_1), (v_2, v'_1)$) indicate the dependency of physical servers on a logical node in terms of the flow of coordination messages from the orchestrator to the physical servers. When the physical server $v_1$ fails, the logical node hosting the orchestrator $v_2$ loses its dependent physical node $v_1$, and becomes nonfunctional. This induces another loss of the dependent node of $v'_1$, and eventually the single node failure causes a failure of the whole network.

Cascading failures can also lead to the malfunctioning of CPSs. In fact, it has been reported that some major electricity outages in smart grids, such as the 2003 nation-wide blackout in Italy \cite{Italyblackout}, and the 2004 blackout over 8 states in US and 2 provinces in Canada \cite{1525122}, were due to cascading failures induced from poorly designed dependencies between the electricity network and control information network.

Many contributions have been made since the first theoretical proposal on the cascading failure model by Buldyrev et al. in 2010 \cite{Buldyrev2010}. The pioneering works \cite{Buldyrev2010,gao2012networks} focus on analyzing the behavior of cascading failures rather than proposing design strategies. In contrast, some following works identify vulnerable topologies in interdependent networks to avoid such fragile structures in the design phase by investigating the relation between node degree and failure impacts \cite{7179465}, or evaluating the importance of nodes exploiting the algebraic expression of dependencies \cite{6849338}. Furthermore, other works propose design strategies in more realistic models to consider the impact of failures caused by a single component \cite{ICC17_ILPDesignDependency}, integrated factors within and between layers \cite{ICNC16_influencemodel}, or the heterogeneity of nodes in each layer \cite{ICNC16_hetero}.

This paper discusses a design problem for interdependent networks to improve their survivability, which is a measure of the robustness against a whole network failure, by modifying an existing network topology. The contribution that contrasts our work with other related works is the consideration of existing network facilities. Our method is aimed at redesigning a relatively small part of the existing network to enhance the survivability so that the entire network remains operational even during the restructuring process. In order to realize this continuous availability, a special type of dependency, whose removal does not influence the functionality of the entire system, is identified in the first step of our restructuring method. Our heuristic algorithm increases the survivability of entire systems by the relocations of these dependencies. 
While our previous work \cite{my_ICC} allows a node to have dependencies with any nodes in the other layer, this paper extends the model by considering geographical, economic, or logical accessibility of provisioning by nodes. These constraints are represented as clusters of nodes, and an interdependent network is modeled as a directed graph consisting of multiple clusters. The membership of a node in a specific cluster imposes restrictions on the nodes to which the node can provide support, and the nodes from which the node can receive support. Hence, possible modifications to the dependencies between nodes would vary, depending on the cluster to which a node belongs. Finally, our method is evaluated by simulations in different pseudo interdependent networks.

\section{Related Works}
Most of the preceding works on interdependent networks attempt to analyze the behavior of cascading failures in well-known random graphs, which have certain characteristics in degree distributions and underlying topology \cite{Buldyrev2010,gao2012networks}. Those works analyze the propagation of failures based on percolation theory developed in the field of random networks. Following the directions shown by a seminal work by Buldyrev et al. in \cite{Buldyrev2010}, more general models are discussed in \cite{gao2012networks}.

The works \cite{7179465, 6849338, ICNC16_influencemodel, ICNC16_hetero, ICC17_ILPDesignDependency} focus on the design aspect of interdependent networks. The relation between the impact of failures and interdependencies is empirically demonstrated to decide appropriate dependency allocations in \cite{7179465}. A method to evaluate the importance of nodes \textcolor{black}{in terms of} network robustness is proposed in \cite{6849338} by introducing a novel representation of interdependencies based on boolean algebra. This evaluation enables network operators to prioritize the protection of the nodes that contribute \textcolor{black}{more} to the robustness of the network. In \cite{ICNC16_influencemodel}, the authors consider dependency relations not only between layers but also within a single-layer. Combining multiple factors that make a node nonfunctional, their method adjusts the dependency of a node on the other nodes. The work in \cite{ICNC16_hetero} also considers the influence within a single-layer, supposing the heterogeneity of nodes. In this model, a network can have different types of nodes such as generating and relay nodes. Zhao et al. \cite{ICC17_ILPDesignDependency} formulate an optimization problem enhancing the system robustness, \textcolor{black}{defining} Shared Failure Group (SFG), a group of nodes that can simultaneously fail due to a cascading failure initiated by the same component.


Another branch of interdependent network research is recovery after failures \cite{STIPPINGER2014481, PhysRevE.92.052806, 7913673, 7440669, 7842042, NaturecommMajdandzic}. The works in \cite{STIPPINGER2014481, PhysRevE.92.052806, 7913673} analyze the behaviors of failure propagations when each node performs local healing, where a functioning node substitutes for the failed node by establishing new connections with its neighbors. The speed of further cascades and resulting network states are revealed by percolation theory \cite{STIPPINGER2014481, PhysRevE.92.052806} or steady state analysis in the belief propagation algorithm \cite{7913673}. Also, resource allocation problems, which consider the different roles of network nodes are discussed in \cite{7440669, 7842042, NaturecommMajdandzic}. The order of assigning repairing resources is a critical problem during the recovery phase when the amount of available resources is limited. The works in \cite{7440669, 7842042} propose node evaluation measurements to decide the allocation, while an equivalent problem in the phase diagram is discussed in \cite{NaturecommMajdandzic}.

Our work proposes a method to improve the survivability of interdependent networks, following the survivability definition in \cite{MITpaper2013Globecom}. Our work would be classified into the category of protection design methods before failures. Specifically, the proposed method is exploited in a redesign process of an existing network to enhance the survivability, while the existing works \cite{7179465, 6849338, ICC17_ILPDesignDependency, ICNC16_influencemodel, ICNC16_hetero} discuss the initial design of an entire network. Our protection method, considering the functionality during the redesign, would reduce the cost of survivability improvement in contrast to the entire reconstruction of the systems.

\section{Modeling and Motivating Example}
In this section, we present a mathematical model for describing interdependent networks, and we present a motivating example of our method. Section \ref{survofinterdep} summarizes related work \cite{MITpaper2013Globecom} defining the survivability for interdependent networks, which we adopt to evaluate the networks.

\subsection{Network Model}
An interdependent network consists of $k$ constituent graphs $G_i = (V_i, E_{ii})\ (1 \leq i \leq k)$ \textcolor{black}{and their} interdependency relationships, which are defined by sets of (directed) arcs $A_{ij}\ (1 \leq i, j \leq k,\ i \neq j)$ representing \textcolor{black}{the provisioning between} a pair of nodes in different graphs. Edges in $E_{ii} \subseteq V_i \times V_i$ are called \textit{intra-}edges because they connect pairs of nodes in the same network. In contrast, arcs in $A_{ij} \subseteq V_i \times V_j\ (i \neq j)$ are called \textit{inter-} or \textit{dependency} arcs. If there exists an arc $(v_i, v_j) \in A_{ij}\ (v_i \in V_i, v_j \in V_j)$, it means that a node $v_j$ has dependency on a node $v_i$. The node $v_i$ is called the \textit{supporting} node, and $v_j$ is a supported node. A node $v$ is said to be \textit{functional} if and only if it has at least one functional supporting node.

When an interdependent network is logically partitioned, each constituent graph $G_i$ has a clustering function $\kappa_i: V_i \longrightarrow \{1, 2, ..., \gamma_i\}$, where $\gamma_i \in \mathbb{N}$ is the number of clusters in $G_i = (V_i, E_{ii})$. Then, a graph $I_i^x = (W_i^x \subseteq V_i, E_{ii}(W_i^x))$ induced by a node set $W_i^x = \{v \mid \kappa_i(v) = x\ (1 \leq x \leq \gamma_i)\}$ is called a \textit{cluster}. Note that this definition insists that a node is in exactly one cluster.

In order to emphasize the dependency between constituent graphs, an interdependent network can be represented as a single-layer directed graph $G = (V, A)$, where $V := \bigcup_{i} V_i$, and $A := \bigcup_{\{(i, j) \mid i \neq j\}} A_{ij}$ by abbreviating intra-edges. With this notation, a node $v$ is said to be \textit{functional} if and only if $\indeg(v) \geq 1$. Note that all the discussions in the rest of this paper follow this single-layer graph representation.

Additionally, we introduce a different notation of arcs with respect to their source nodes. Let $A(v) \subseteq A$ represent a set of arcs whose source node is $v \in V$. To identify each arc during the restructuring process, where some arc temporarily loses its destination, each arc is denoted as $(v, \cdot)_m\ (m = 1, ..., \outdeg(v))$. The index $m$ is a given fixed identification number for each arc in $A(v)$. Hence, every arc in $A$ can be specified by providing source node $v$ and its identification number $m$.

A set of constituent graphs is totally ordered by the number of nodes that are the source of at least one intra-arc: $|V_i^{\mathrm{out}}|$, where $V_i^{\mathrm{out}} := \{v \in V_i \mid A(v) > 0\}$. A constituent graph that has the least number of nodes with outgoing arcs is named the minimum supporting constituent graph $G_i$: $|V_i^{\mathrm{out}}| \leq \min_{j} |V_j^{\mathrm{out}}|$.

\subsection{Survivability of Interdependent Networks \label{survofinterdep}}
Parandehgheibi et al. \cite{MITpaper2013Globecom} propose an index that quantifies the survivability of interdependent networks against cascading failures exploiting the \textit{cycle hitting set}, and they prove that the computation of the survivability is NP-complete. They show that a graph needs to have at least one directed cycle in order to maintain some functional nodes; in other words, the existence of one cycle prevents an interdependent network from its entire failure. Thus, the survivability of interdependent networks is defined as the cardinality of the minimum cycle hitting set whose removal brings non-functionality for the entire network. Note that a cycle hitting set $S$ is a set of nodes such that any cycle $C = (V(C), E(C))$ in a given graph $G = (V, A)$ has at least one node in the hitting set: $S \cap V(C) \neq \emptyset, \ \forall C \in \mathcal{C}(G)$, where $\mathcal{C}(G)$ is the set of all cycles in the given graph. This definition implies that the entire failure of an interdependent network occurs when the corresponding graph becomes acyclic. Let $H(G)$ denote a cycle hitting set with the minimum cardinality: $|H(G)| := \min_{S \in \mathcal{S}} |S|$, where $\mathcal{S}$ is the set of all the cycle hitting sets in $G$. Formally, the survivability of an interdependent network $G$ is the cardinality of the minimum cycle hitting set, $|H(G)|$.

\def\scale{0.85}
\begin{figure}[t]
\centering
\begin{minipage}{0.47\columnwidth}
\centering
\includegraphics[width=\scale\textwidth]{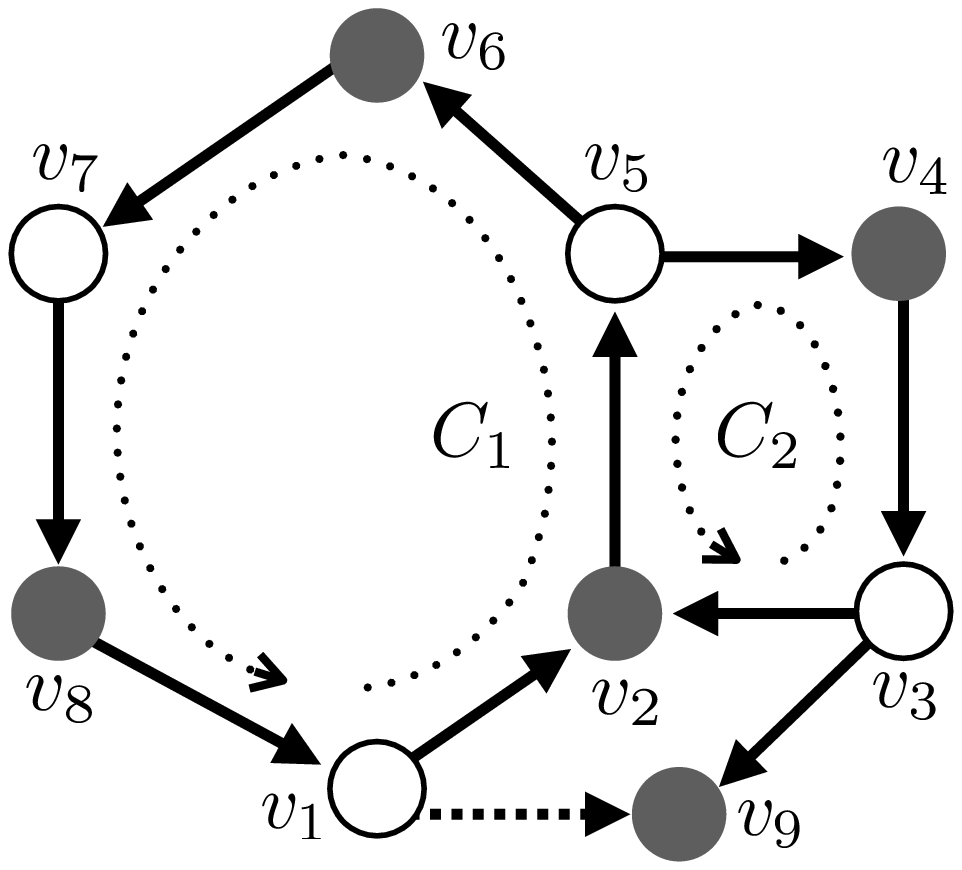}
\caption{Graph $G$ with $(v_1, v_9)$. \label{original1}}
\end{minipage}\hfill
\begin{minipage}{0.47\columnwidth}
\centering
\includegraphics[width=\scale\textwidth]{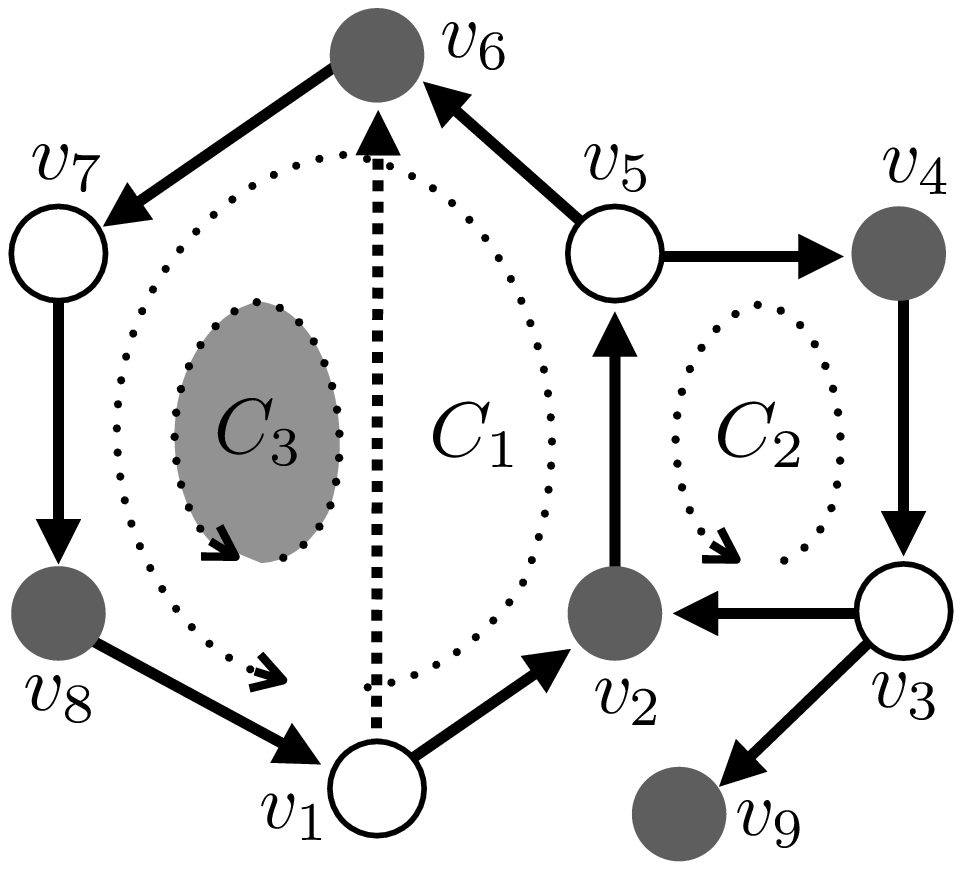}
\caption{Graph $G'$ with $(v_1, v_6)$. \label{modified1}}
\end{minipage}
\end{figure}

\subsection{Motivating Example \label{motivating}}
Adopting the survivability definition shown above, improving survivability would be equivalent to increasing the number of disjoint cycles in a graph. Figs. \ref{original1} and \ref{modified1} show an example comparing two similar interdependent networks.

In graph $G$ in Fig. \ref{original1}, there exists two cycles: $C_1$ and $C_2$. If $v_2$, which is in both $V(C_1)$ and $V(C_2)$, becomes nonfunctional because of a failure, all the nodes in $G$ eventually lose their supporting nodes and become nonfunctional: $H(G) = \{v_2\}$. On the other hand, no single node failure can destroy all the three cycles in $G'$ in Fig. \ref{modified1}, while a two-node failure can make it acyclic (e.g. $H(G') = \{v_2, v_7\}$). Therefore, the graph $G'$ is more survivable than $G$, since $1 = |H(G)| < |H(G')| = 2$, although they differ only in the destination node of one dependency arc ($(v_1, v_9)$ in $G$ or $(v_1, v_6)$ in $G'$). Supposing that $G$ is an existing topology of a network, a method that relocates $(v_1, v_9)$ to $(v_1, v_6)$ can achieve an enhancement of the survivability.

\section{Problem Formulation}
\subsection{Assumptions \label{assumptions}}

This paper deals with the case in which interdependent networks have two types of homogeneous constituent networks with identical dependencies ($k = 2$). However, our discussion with the restriction on $k$ can be easily extended to more general cases. In more advanced network models, each constituent network can have different types of nodes, such as independently functional generating nodes and relay nodes, which need provisioning from a generating node via paths of intra-edges \cite{ICNC16_hetero}. Nevertheless, for simplicity, this work follows the assumption in \cite{MITpaper2013Globecom} that each node in a constituent network is directly connected to a \textcolor{black}{reliable} conceptual generating node \textcolor{black}{by a reliable edge} (homogeneous constituent graphs). Moreover, it is assumed that each supporting node provides a unit amount of support that is enough for a supported node to be operational (identical dependencies), following the same model in \cite{MITpaper2013Globecom}.

Additionally, this paper presumes that each cluster $x$ receives some support from at least one of the clusters that are supported by cluster $x$. In other words, this presumption excludes the case that a cluster does not receive provisions from any of the clusters that the cluster is supporting.

\subsection{Requirement Specification \label{spec}}
One aspect contrasting our scheme to other works is the consideration to improve the survivability of existing interdependent networks by changing some topological structures.
Because all the nodes need to remain functional even during the relocations of dependency relations, it is necessary to avoid the loss of all supporting nodes for any node at any stage of the restructuring. In other words, each node needs to be survivable from a cascading failure, which requires the direct or indirect support by the nodes in directed cycles. This constraint is formally represented as the following rule for the live restructuring.
\begin{enumerate}
\item Every node remains reachable from a node in a directed cycle via at least one directed path at any stage of the restructuring. \label{incom_rule}
\end{enumerate}

In addition to guaranteeing the continuous availability, the amount of provisioning provided by each supporting node should remain the same after the restructuring in order to consider the capability of each node. The capability could be, for example, the limit on electricity generation, computation performance, or the number of ports available. 
\begin{enumerate}
\setcounter{enumi}{1}
\item The number of supports that a node provides must remain less than or equal to its original provisioning capability. \label{outgo_rule}
\end{enumerate}

Furthermore, depending on which cluster a node in graph $G_i$ belongs to, the node has a constraint on clusters in $G_j$ that it can support. The constraint is given by a supportability function $\sigma_{ij}: V_i \longrightarrow 2^{\gamma_j}$, where $2^{\gamma_j}$ is the power set of the cluster indices in a constituent network $G_j$. This means that a node $v\,  (\in V_i)$ can provide its support to the nodes in the clusters of $G_j$ given by the supportability function. This specification corresponds the geographical, economic, or logical constraints on the accessibility of supports from a node to specific groups of nodes. For example, it is impossible for information control node $v$ to have electricity supply from node $u$ if $v$ and $u$ are geographically far apart or managed by different administrative institutions. The geographical or administrative domain is shown as a cluster in each constituent graph, and dependency relations of the nodes should be closed within a set of permitted nodes, which are geographically close, or managed by the same company or allied companies, since each cluster should be independent from the outsiders. This constraint relating to network clustering is simply expressed as follows.
\begin{enumerate}
\setcounter{enumi}{2}
\item All the provisionings from a node $u$ are directed towards the nodes in the clusters that $u$ can support, as designated by the supportability function $\sigma_{ij}$.
\end{enumerate}

\subsection{Clustered $\Delta H$ Problem}
This section formulates the clustered $\Delta H$ problem, which is aimed at enhancing the survivability of a given interdependent network with clusters by restructuring dependency relationships, considering the continuous availability, supporting capability, and clustering constraint of each node. 

\textcolor{black}{Considering the continuous availability} of an existing network during restructuring leads to the formulation of a gradual reconstruction problem, where no relocation of two or more different arcs is conducted at a time. Each phase relocating one arc is named a \textit{step}. Let $G^s = (V, A^s)$ denote the graph representing the interdependent network topology at step $s$. The improved interdependent network $G^{s+1}$ after step $s$ consists of a node set $V$, \textcolor{black}{which is the same node set as in graph} $G^s$, and an arc set $A^{s+1}$ amended by the relocation of an arc $(u, v) \in A^s$ to $(u, v')$, where $v' \in V$ is a new destination for the arc $(u, v)$.

The clustered $\Delta H$ problem is to maximize the difference \textcolor{black}{in survivability} between a given interdependent network, which is recognized as $G^0$, and the resulting network after a sequence of consecutive improvements. The resulting network is represented as $G^f$, where $f$ denotes the step at which the last arc relocation is completed. Formally, the objective is to maximize the difference between $|H(G^0)|$ and $|H(G^f)|$, which is defined as $\Delta H$.

\begin{problem*}[Clustered $\Delta H$ Problem]
For a given $G^0 = (V = \bigcup_i V_i, A^0)$, the number of clusters $\gamma_i \in \mathbb{N}$ in each constituent graph $G_i$, a clustering function $\kappa_i: V_i \longrightarrow \{1, 2, ..., \gamma_i\}$ for each constituent graph $G_i$, and supportability functions $\sigma_{ij}: V_i \longrightarrow 2^{\gamma_j}$, maximize $\Delta H := |H(G^f)| - |H(G^0)|$, where $G^{s+1} = (V, A^{s+1})\ (0 \leq s \leq f-1)$ is obtained by the relocation of the destination of a single arc in $A^{s}$: $A^{s+1} = A^{s} \setminus (u, v) \cup (u, v')$, satisfying 
\begin{enumerate}
\item $\indeg(v)_{G^s} \geq 1\ \ \forall v \in V$, \label{cond2}
\item $\outdeg(v)_{G^{s+1}} = \outdeg(v)_{G^s}\ \ \forall v \in V$, \label{cond1}
\item $\kappa_j(v \in V_j) \in \sigma_{ij}(u \in V_i)\ \ \forall (u, v) \in A^s$. \label{cond3}
\end{enumerate}
\end{problem*}

These three conditions correspond to the three rules described in Section \ref{spec}. The second and third conditions are easily derived from the corresponding rules. Lemma \ref{incom_is_enough} shows the equivalence of the condition \ref{cond2} and Rule \ref{incom_rule}.

\begin{lemma}
When $\indeg(v)_G \geq 1\ (\forall v \in V)$ in a connected directed graph $G = (V, A)$, (a) $G$ has at least one directed cycle, and (b) any node $v \in V$ is reachable from a node $u \in V$ that is contained in a directed cycle.
\label{incom_is_enough}
\end{lemma}
\begin{proof}[Proof] 
$\indeg(v)_G \geq 1\ (\forall v \in V)$ insists that any node $v$ has at least one parent $v'$. The path $v \leftarrow v' \leftarrow ...$ composed by repeating the trace of parents can be acyclic until the length of the path is $|V - 1|$. However, the $|V|$th node must have at least one parent from the assumption. Thus, the pigeonhole principle indicates that it is necessary that the path forms a directed cycle.
\end{proof}

\subsection{Problem Analysis}
This section provides the analysis on the trivial optimal case of the clustered $\Delta H$ problem with a special setting, where each of constituent graph consists only of one cluster. Let $\rho((u, \cdot)_m)$ denote the number of relocations that arc $(u, \cdot)_m \in A$ experienced during the restructuring process. Note that $\sum_{u \in V} \sum_{m=1}^{\outdeg(u)} \rho((u, \cdot)_m) = f.$ 

From the definition, the optimum survivability cannot exceed the number of supporting nodes, \textcolor{black}{which each have} at least one outgoing arc, in the minimum supporting constituent graph $G_i$. This is because a set of such nodes covers all the directed cycles in an interdependent network $G$. This observation implies that the optimum survivability is achieved when every node $v_i \in V_i$ of $G_i$ has an injective mapping to a node in $V_j\ (j \neq i)$. In other words, for each node $v_i$ in $G_i$, there exists at least one unique disjoint cycle whose length is 2 with $v_j$ in $G_j$. The following lemma gives a sufficient condition to reach the ideal state by repeated relocations while preserving the problem constraints.

\begin{figure}[t]
\centering
\begin{minipage}[t]{0.48\columnwidth}
\centering
\includegraphics[width=0.9\columnwidth]{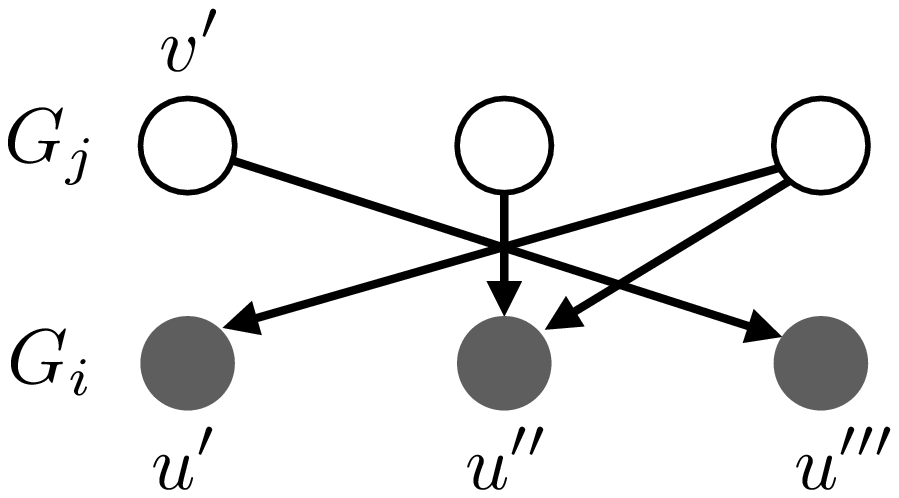}
\caption{Original Dependencies, where $(v', u')$ is missing. Note that this figure only shows $A_{ji}$. The symmetric discussion can be done for $A_{ij}$. \label{lemma21}}
\end{minipage}\hfill
\begin{minipage}[t]{0.48\columnwidth}
\centering
\includegraphics[width=0.9\columnwidth]{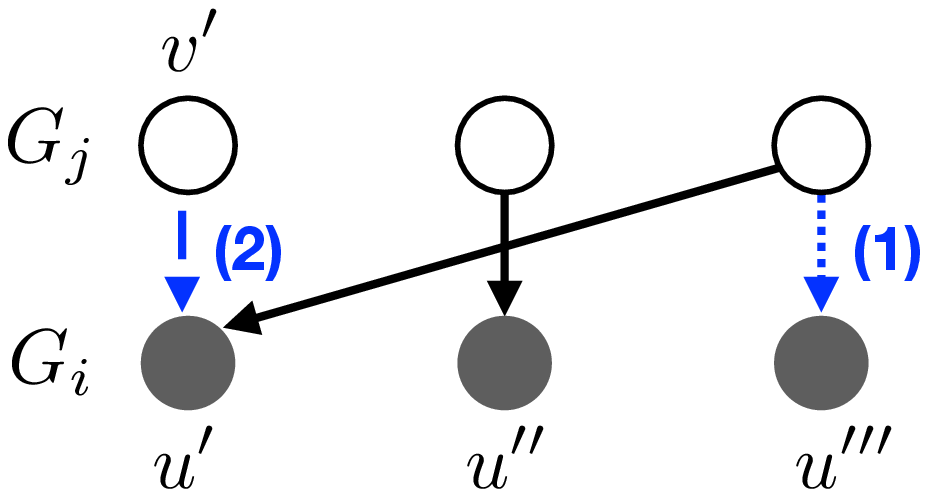}
\caption{Relocation Steps (1) to maintain the functionality of $u'''$, and (2) to form a length-2 cycle with $v'$ and $u'$. \label{lemma22}}
\end{minipage}
\end{figure}

\begin{lemma}
When the number of relocations for each arc $\rho((u, \cdot)_m)$ is not upper bounded, \textcolor{black}{in order to have the optimum restructuring, it is sufficient that} the minimum supporting constituent graph $G_i$ satisfies $|V_j| < \sum_{u \in V_i} |A(u)|$ and $\sum_{v \in V_j} |A(v)| > |V_i|\ (j \neq i)$. Then, the optimum survivability becomes $|V_i^{\mathrm{out}}|$.
\label{length2cycles}
\end{lemma}
\begin{proof}
The maximum survivability achievable by restructuring is equal to the number of nodes \textcolor{black}{that have at least one outgoing arc} $|V_i^{\mathrm{out}}|$ in the minimum supporting constituent graph $G_i = (V_i, E_{ii})$, because the removal of such nodes from $G_i$ must destroy all the cycles between $G_i$ and another constituent graph. \textcolor{black}{In order to achieve the maximum survivability via the restructuring process}, it is necessary that each node $u \in V_i^{\mathrm{out}}$ belongs to a cycle whose length is 2. Otherwise, the cycle contains another node $w \in V_i^{\mathrm{out}}$, and the removals of such $w$'s make $u$ lose all incoming arcs. Note that a node in $V_i \setminus V_i^{\mathrm{out}}$ is never a part of directed cycles, since it has no outgoing arc.

Suppose that we have the minimum supporting constituent graph $G_i$ and another constituent graph $G_j$ that satisfy the two conditions in the lemma. From the definition of the minimum supporting constituent graph, we can make $|V_i^{\mathrm{out}}|$ pairs of nodes \textcolor{black}{$\langle u \in V^{\mathrm{out}}_i, v \in V_j^{\mathrm{out}} \rangle$}, which are expected to form a length-2 cycle together after restructuring, so that no \textcolor{black}{two nodes in $V_i$ are paired} with the same node in $V_j^{\mathrm{out}}$. 

Figs. \ref{lemma21} and \ref{lemma22} illustrate a general example of a restructuring process to form such a length-2 cycle by dependency arc relocations. Note that the figures only show $A_{ji}$, but the symmetric argument can be done for $A_{ij}$. Let \textcolor{black}{$\langle u' \in V^{\mathrm{out}}_i, v' \in V_j^{\mathrm{out}} \rangle$} be a pair such that $(v', u') \notin A_{ji}$. \textcolor{black}{In order to make a length-2 cycle between $v'$ and $u'$, the arc $(v', u''')$ should be relocated to $(v', u')$. However, the relocation makes $u'''$ lose all of its incoming arc. The loss of incoming arc of $u'''$ is always avoided by relocating one of the arcs incoming to $u''$ to $u'''$ (See Figs. \ref{lemma21} and \ref{lemma22} (1)). The supposition in the lemma and the pigeonhole principle suggest the existence of at least one node $u'' \in V_i$ that has two incoming arcs. After the adjustment of the provisioning for $u'''$ by this relocation, the arc $(v', u''')$ can be relocated to $(v', u')$ (See \ref{lemma21} and \ref{lemma22} (2)).}  

For a pair $\langle u' \in V^{\mathrm{out}}_i, v' \in V_j^{\mathrm{out}} \rangle$ such that $(u', v') \in A_{ij}$, similar relocations are always possible, because $|V_j| < \sum_{u \in V_i} |A(u)|$. Thus, these relocations eventually achieve the maximum survivability by forming $|V_i^\mathrm{out}|$ length-2 cycles that each consist of a pair $\langle u \in V^{\mathrm{out}}_i, v \in V_j^{\mathrm{out}} \rangle$.
\end{proof}

Some propositions similar to Lemma \ref{length2cycles} appear in related literature \cite{ICC17_ILPDesignDependency, 7935429}. The sufficient condition provided in Lemma \ref{length2cycles} allows the entire restructuring of inter-arcs by repeated relocations of each arc. Therefore, the $\Delta H$ problem is recognized as a design problem of an entire interdependent network discussed in \cite{ICC17_ILPDesignDependency} under these assumptions. Also, the work \cite{7935429} claims that such a one-to-one provisioning relation realizes the robustness, while assuming certain structural characteristics of random graphs.

However, it is unrealistic to relocate a dependency arc many times, when considering the overhead of the changes of provisioning relations in network systems. Therefore, the following part of our paper discusses the case where the number of relocations are strictly restricted: $\rho((u, \cdot)_m) \leq 1\ (1 \leq m \leq \outdeg(u),\ \forall u \in V)$. Under this condition, it cannot be guaranteed to obtain the optimum survivability even when the sufficient condition above holds.

\begin{figure*}[t]
\centering
\begin{minipage}[t]{0.3\textwidth}
\centering
\includegraphics[width=0.75\textwidth]{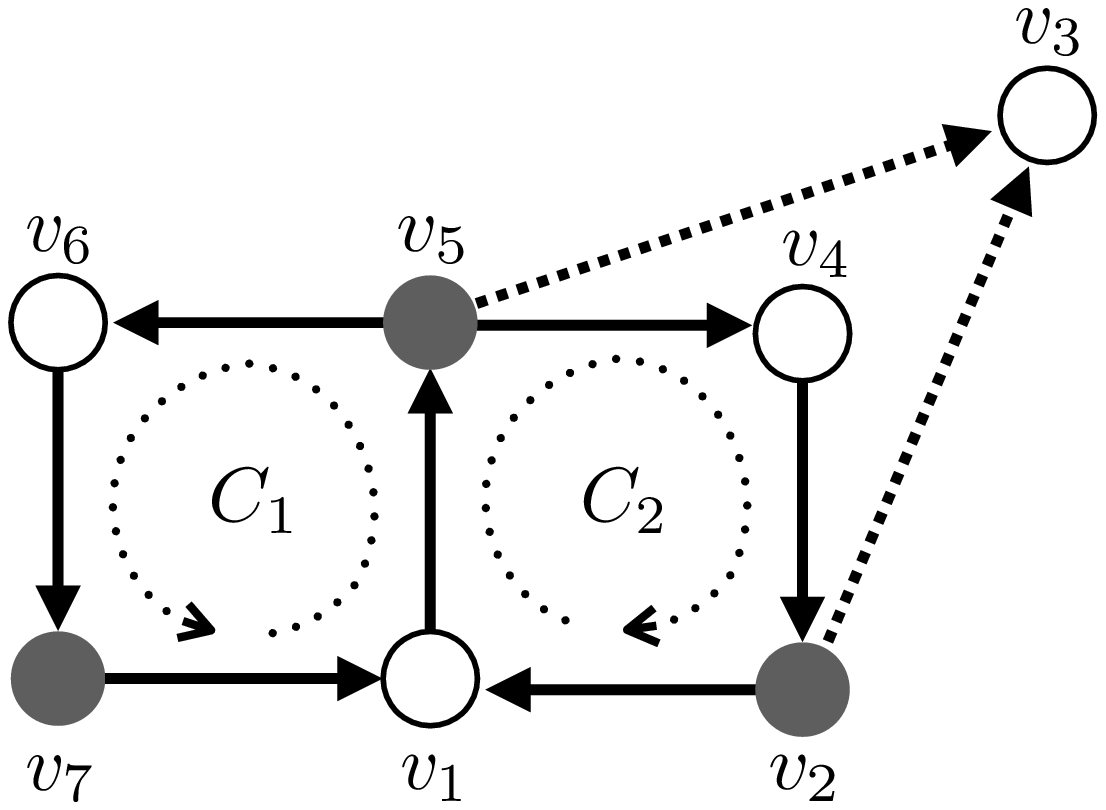}
\caption{Original graph $G$ with Marginal Arcs $(v_2, v_3)$ and $(v_5, v_3)$. \label{original2}}
\end{minipage}\hfill
\begin{minipage}[t]{0.3\textwidth}
\centering
\includegraphics[width=0.75\textwidth]{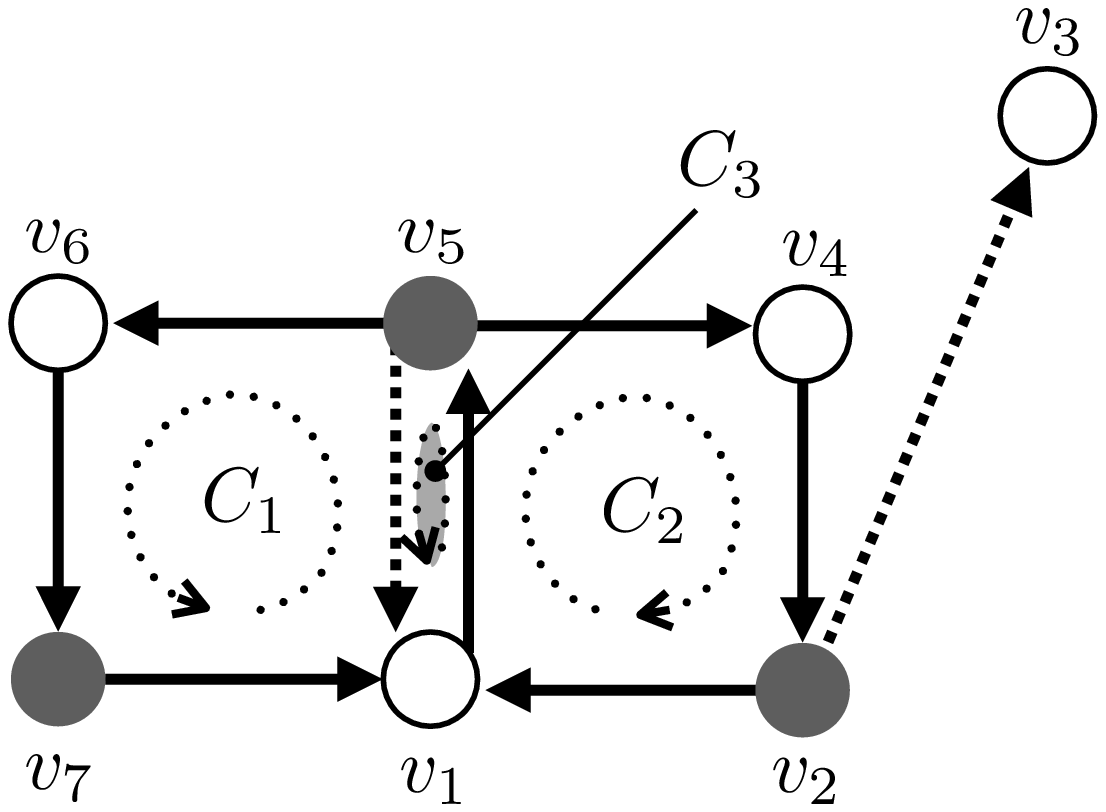}
\caption{Modified graph $G'$ with a new arc $(v_5, v_1)$. \label{modified2_1}}
\end{minipage}\hfill
\begin{minipage}[t]{0.3\textwidth}
\centering
\includegraphics[width=0.75\textwidth]{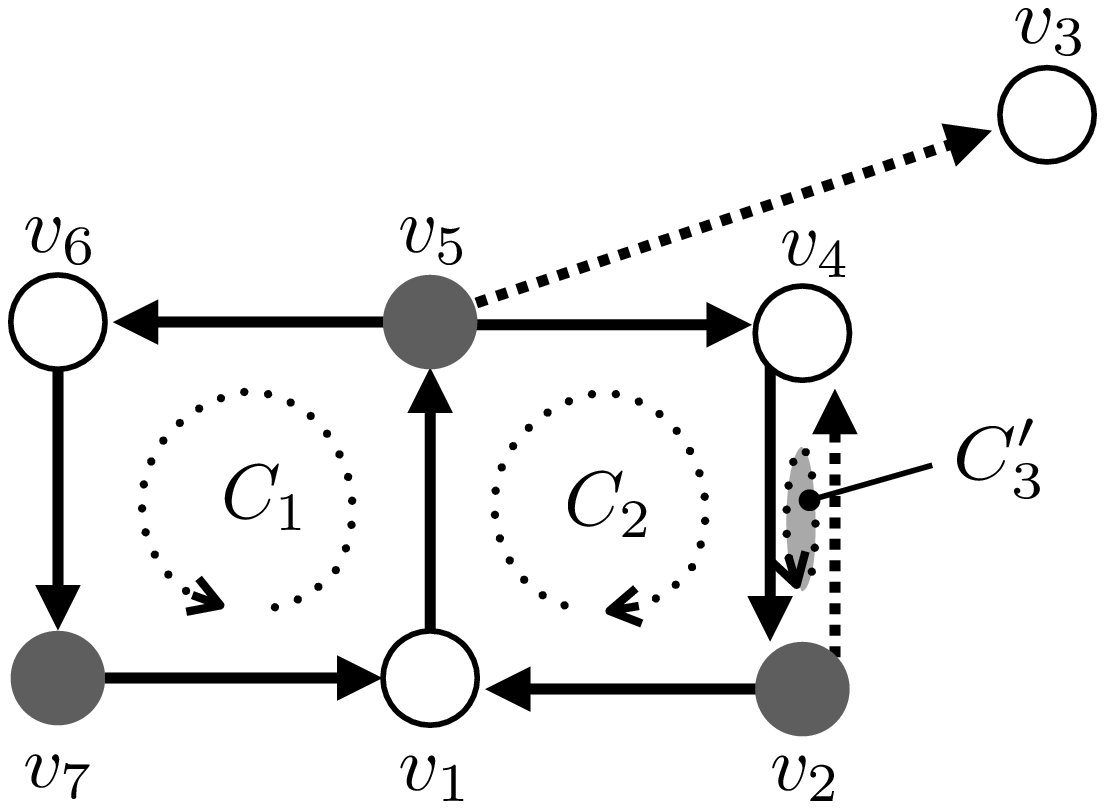}
\caption{Modified graph $G''$ with a new arc $(v_2, v_4)$. \label{modified2_2}}
\end{minipage}
\end{figure*}

\section{heuristic algorithm for $\Delta H$ Problem}
This section proposes a heuristic algorithm for the clustered $\Delta H$ problem. Before providing the details of our heuristic algorithm, we first define special types of arcs named Marginal Arcs (MAs), which are candidates for the relocations in Section \ref{restructure_benefit}. Then, the heuristic algorithm, which consists of two algorithms: Find-MAs and $\Delta H$, is described. The Find-MAs algorithm enumerates all the arcs that match the definition of MAs. With the set of MAs found by the Find-MAs algorithm, the $\Delta H$ algorithm decides appropriate relocations of the dependency arcs in the set, considering disjointness of newly formed cycles, so that it can improve the survivability of a given network. 

After the discussion for a simple case with only one cluster in each constituent graph in Sections \ref{findMAdesc} to \ref{deltasec}, Section \ref{clustercase} explains how the other cases with multiple clusters are broken down into the simple case.

\subsection{Restructuring of Dependencies \label{restructure_benefit}}
In order to guarantee continuous availability, it is necessary to classify the dependency arcs into either changeable or fixed arcs. However, it is computationally difficult to know the classification beforehand under the condition of $\rho((u, \cdot)_m) \leq 1\ (\forall u \in V)$, because this process involves enumeration of all the permutations of arc relocations and their combinations of destinations. Thus, in this paper, the classification is simplified by using a sufficient condition, while this enumeration is likely to become another optimization problem for a further investigation.

As observed in Section \ref{motivating}, increasing disjoint cycles in a given network could be an important factor to enhance overall survivability. Hence, our method maintains all existing cycles, which is sufficient to avoid cascading failures, and tries to reallocate the destinations of the arcs that do not belong to directed cycles and that do not make their descendant nodes nonfunctional. Let the arcs that are not in any cycles in a given directed graph $G = (V, A)$ be called \textit{Marginal Arcs} (MAs). Formally, the set $M \subsetneq A$ of MAs is defined as
\begin{align}
M := \{(u, v) \mid (u, v) \notin A(C)\ \forall C \in \mathcal{C}(G) \}.
\label{MAdef}
\end{align}

\begin{lemma}
A removal of any marginal arc never decreases the survivability of an interdependent network: $|H(G)| \leq |H(\overline{G})|$, where $G$ is a given graph, and $\overline{G}$ is the graph obtained by the removal.
\label{doesntdecrease}
\end{lemma}

\begin{proof}[Proof]
Let $M$ be a set of marginal arcs. From the definition of MAs (Eq. \eqref{MAdef}), the removal of MAs does not destroy or connect any existing cycles in $G = (V, A)$. Therefore, $|H(G)| = |H(\overline{G})|$, where $\overline{G} = (V, A \setminus M)$.
\end{proof}

Moreover, appropriate relocations of the removed MAs could improve the survivability of interdependent  networks, assuring operability during the relocation process and maintaining the provisioning capability of each node. Let us analyze the effect of dependency relocations using simple examples in Figs. \ref{original2}-\ref{modified2_2}. The given graph $G$ in Fig. \ref{original2} has two marginal arcs: $M = \{(v_2, v_3), (v_5, v_3)\}$. In order to maintain at least one supporting node for $v_3$, one of the MAs has to remain the same, and the other can be relocated. Fig. \ref{modified2_1} shows the case of relocating $(v_5, v_3)$ to $(v_5, v_1)$; on the other hand, Fig. \ref{modified2_2} indicates the case of relocation of $(v_2, v_3)$ to $(v_2, v_4)$. Even though one new cycle ($C_3$ and $C'_3$ respectively) is formed by each relocation, the modified graphs $G'$ and $G''$ have different survivability: $|H(G')| = 1\ (= H(G))$, and $|H(G'')| = 2$. This is because the cycles in $G'$ are not disjoint with each other: $V(C_1) \cap V(C_2) \cap V(C'_3) \neq \emptyset$; in contrast, $V(C_1) \cap V(C_2) \cap V(C''_3) = \emptyset$ in $G''$. Therefore, it could be said that the appropriate relocation for improving survivability is to form disjoint cycles.

\begin{algorithm}[t]
\caption{$\Delta H$-algorithm$(G, l)$}
\label{MAimprove}
\begin{algorithmic}[1]
\REQUIRE subgraph (directed graph) $G = (V, A)$, maximum hop $l \in \mathbb{N}$ (odd)
\STATE $M \gets$ find-MAs($G$)  \hspace{2em} \# $M \subset A$
\FOR{each $(v, w) \in M$}
	\IF{$\textrm{deg}_{\textrm{in}}(w) \geq 1$ after $A \setminus \{(v, w)\}$}
		\WHILE{True}
			\STATE pick $C \in \mathcal{C}(v)$ (randomly)
			\FOR{$i \gets l$; $i > 0$; $i \gets i - 2$}
				\STATE pick $u \in V(C) : \overline{d_{C}}(v, u) = i$
				\IF{$u \notin U$}
					\STATE $A \gets A \setminus (v, w) \cup (v, u)$
					\STATE $U \gets U \cup \{n \mid \overline{d_{C}}(v, n) \leq i\}$
					\STATE break to next arc in $M$ (line 2)
				\ENDIF
			\ENDFOR
		\ENDWHILE
		\STATE pick $(u, v) \in \inarcs(v)$ (randomly) \# Minimal-add process (line 15,16)
		\STATE $A \gets A \setminus (v, w) \cup (v, u)$ \hspace{2em} 
	\ENDIF
\ENDFOR
\end{algorithmic}
\end{algorithm}

\subsection{Find-MAs Algorithm \label{findMAdesc}}
The Find-MAs algorithm first distinguishes MAs $M$, which are candidate arcs for relocations, from the arcs in directed cycles in a given graph $G = (V, A)$, by employing Johnson's algorithm \cite{JognsonCycles}. Johnson's algorithm enumerates all elementary cycles in a directed graph within $O((|V| + |E|)(|\mathcal{C}(G)| + 1))$. It is enough for distinguishing MAs to obtain elementary directed cycles because any non-elementary cycle can be divided into multiple elementary cycles within which dependency relationship are closed. After the enumeration of cycles in $G$ by Johnson's algorithm, the set of MAs is obtained by $M \gets A \setminus \bigcup_{C \in \mathcal{C}(G)} A(C)$.

\subsection{$\Delta H$ Algorithm \label{deltasec}}
With the set of MAs obtained by Johnson's algorithm, the $\Delta H$ algorithm (shown as pseudo code in Algorithm \ref{MAimprove}) relocates the destinations of MAs, considering disjointness of newly created cycles. (See the discussion in Section \ref{restructure_benefit}.) For each MA $(v, w)$, our algorithm first checks whether or not the relocation of this MA causes the loss of supports for the current destination $w$: $\indeg(w)_{\overline{G} = (V, A \setminus \{(v, w)\})} \geq 1$ (line 3).

If $w$ still has some supporting node after the removal of $(v, w)$, the next step is determining a new destination for $(v, \, \cdot\, )$. Our algorithm randomly selects one of the cycles that contains the source $v$ denoted by $C \in \mathcal{C}(v)$ (line 5). There \textcolor{black}{may be multiple} possible candidate nodes for a new destination in the cycle $C$. Thus, the new destination is decided by the size of the newly formed cycle, which is a result of the relocation (line 6, 7). To represent the size of the newly formed cycle, the distance from a node $v$ to a node $u$ in an (existing) cycle $C$ in the counter direction is denoted as $\overline{d_C}(v, u)$ in our pseudo code. When the maximum hop is designated by $l$, the algorithm tries to make a new cycle with size $l + 1$ using a node $u$, such that $\overline{d_C}(v, u) = l$, as the destination of the MA. If it fails to form the cycle, \textcolor{black}{it attempts to compose a smaller cycle using a node $u'$ such that $\overline{d_C}(v, u') = l - 2$. Because of the definition of the dependency, an arc must span between two different layers or constituent networks. Since the node at $\overline{d_C}(v, u) = l - 1$ in $C$ is in the same constituent network as the source node $v$, it cannot be a new destination}.

\textcolor{black}{Consider} an example using a given graph $G$ shown in Fig. \ref{original1} and the restructured graph in Fig. \ref{modified1}. Since the removal of $(v_1, v_9)$ does not make $v_9$ lose all its incoming dependency arcs, our algorithm tries to relocate the destination of this arc to one of the nodes in the cycle $C_1$, which are $v_2, v_6, v_8$. \textcolor{black}{For instance, in the case $l = 3$, a new cycle $C_3$ is formed as depicted in Fig. \ref{modified1} by choosing $v_6$, that satisfies $\overline{d_{C_1}}(v_1, v_6) = l\, (= 3)$. Similarly, if $l$ is initialized to $1$, a new cycle $C_3$ is formed using $\{v_1, v_8\}$.}

\def\scale{0.7}
\begin{figure}[t]
\def\TCOMscale{0.4}
\centering
\begin{minipage}[b]{\columnwidth}
\centering
\includegraphics[width=0.6\columnwidth]{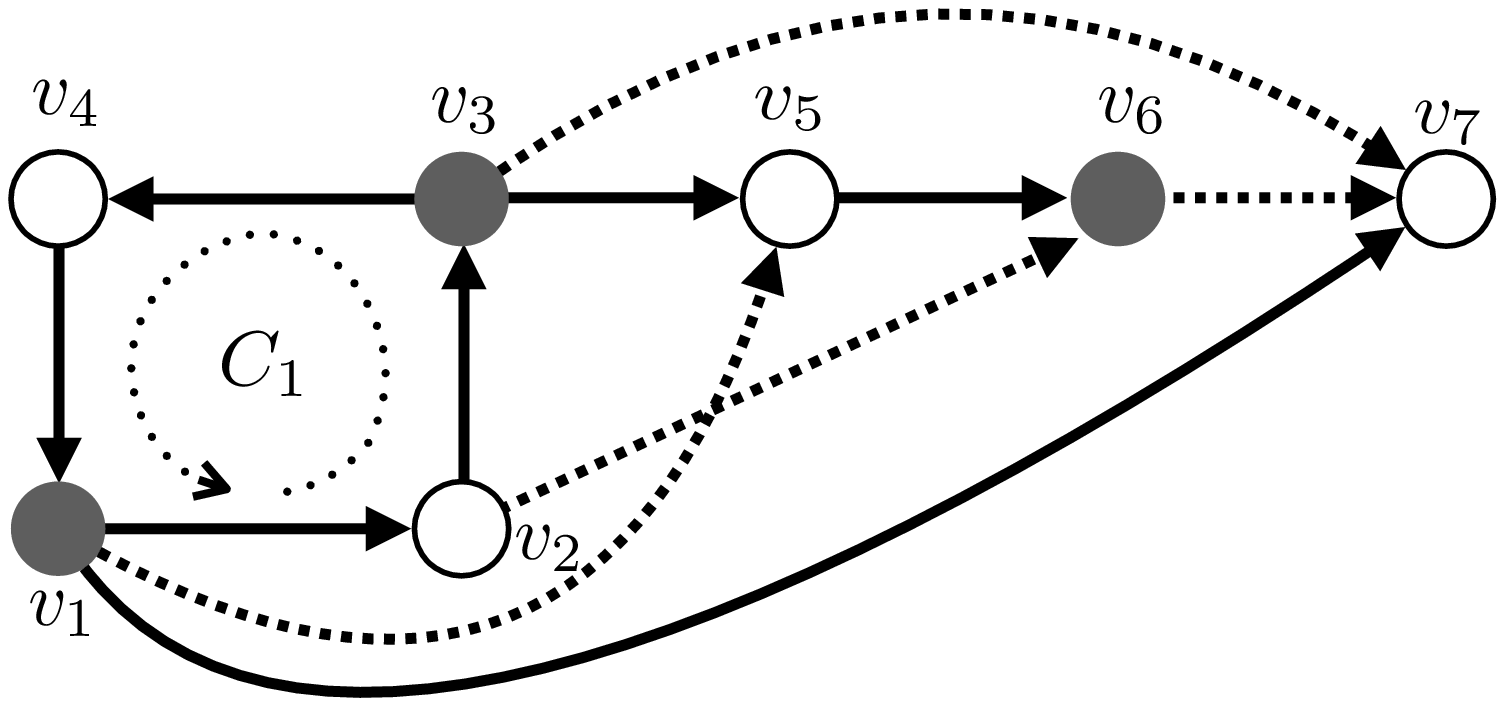}
\caption{A given graph $G$ with $M = \{(v_1, v_5),$$\, (v_2, v_6)$, $(v_3, v_7)$, $(v_6, v_7),$ $(v_1, v_7), (v_3, v_5), (v_5, v_6)\}$. \label{minaddex1}}
\end{minipage}\\ \vspace{1em}
\begin{minipage}[b]{\columnwidth}
\centering
\includegraphics[width=0.65\columnwidth]{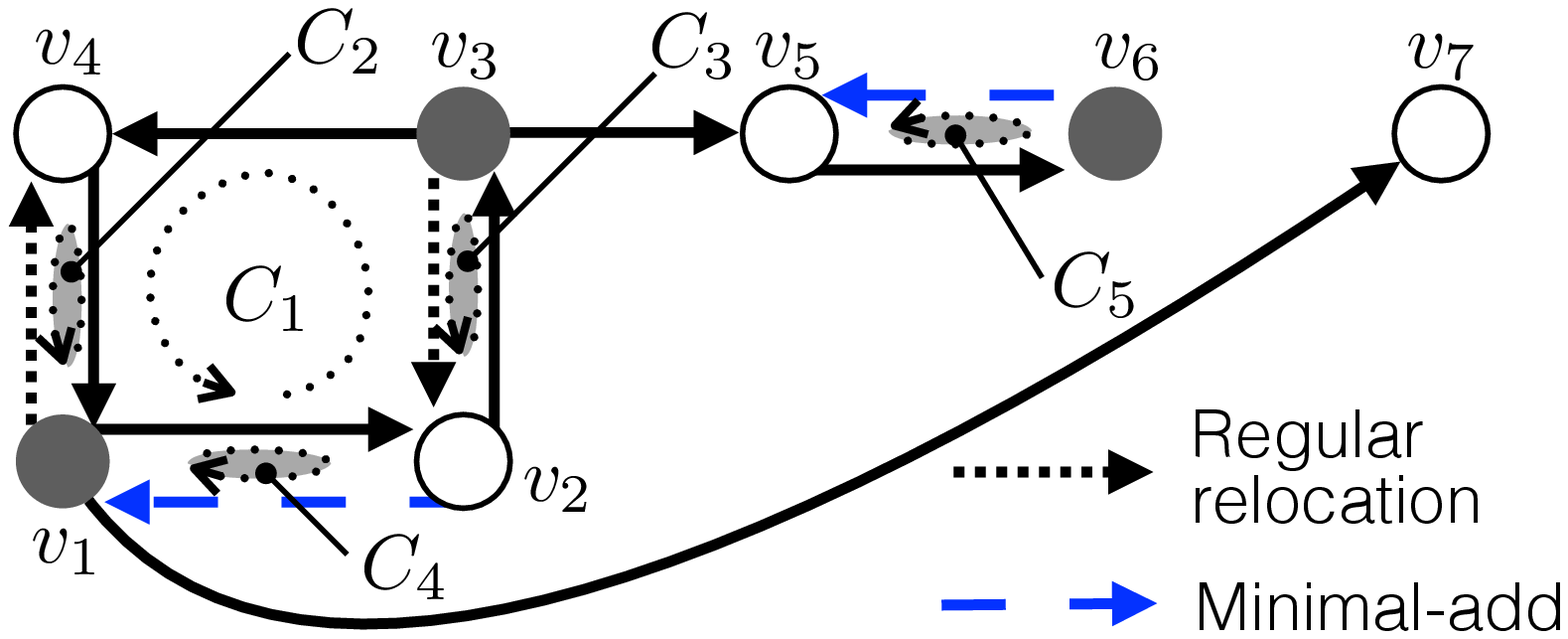}
\caption{A modified graph $G'$ with new arcs: $(v_1, v_4),$$\, (v_3, v_2),$$\, (v_2, v_1),$ $(v_6, v_5)$. \label{minaddex2}}
\end{minipage}
\end{figure}

After selecting a destination candidate $u$ in line 7, our algorithm checks if $u$ is already used to create a new cycle (line 8). This is confirmed by a set of nodes $U$ storing all the nodes that are in newly formed cycles: $\{ n \mid \overline{d_C}(v, n) \leq i\}$ (line 10). For instance in Fig. \ref{modified1}, $U \gets U \cup \{v_1, v_6, v_7, v_8\}$. As will be understood, when another MA tries to form a new cycle using one of these nodes in $U$, the new cycle and $C_3$ share some nodes, which means that those cycles are not disjoint. Also, the arc set $A$ is updated when the new destination is finally fixed (line 9). 

If there exists no possible destination for an MA $(v, w)$ that satisfies all the conditions, the relocation of the MA is conducted by randomly selecting an incoming arc of $v$, $(u, v)$ and relocating $(v, w)$ to $(v, u)$, so that it composes a cycle of length 2 (line 15, 16). This random selection is named \textit{Minimal-add} process.

The MAs relocated by the Minimal-add process satisfy either of the following cases: 1) The node $v$ does not belong to any cycles: $\mathcal{C}(v) = \emptyset$, or 2) all the nodes in the cycles of $\mathcal{C}(v)$ are already used to compose new cycles by other MAs. Figs. \ref{minaddex1} and \ref{minaddex2} show examples of these two conditions (dashed arcs). A given graph $G$ has the MA set $M = \{(v_1, v_5),$$\, (v_2, v_6),$$\, (v_3, v_7),$$\, (v_6, v_7),$ $(v_1, v_7),$$\, (v_3, v_5), (v_5, v_6)\}$. Eventually, the $\Delta H$ algorithm respectively relocates $(v_1, v_5)$ and $(v_3, v_7)$ to $(v_1, v_4)$ and $(v_3, v_2)$. Because $v_6$ is not in any cycles in $G$ (reason 1), the Minimal-add process picks the source of one of the current incoming arcs in $\inarcs(v)$, $v_5$ as the new destination. Also, $(v_2, v_6)$ does not have any possible destinations that are not in the set $U$ (reason 2), and it is relocated to $(v_2, v_1)$ by the Minimal-add process.

\subsection{Application to Clustered Networks \label{clustercase}}
Our heuristic algorithm employs another algorithm named \textit{Decompose-cluster} to form subgraphs, which indicate candidate destinations for the MAs in each cluster, from a given interdependent network. When interdependent networks are clustered, the modification of the destinations of MAs needs to be conducted under more constraints given by supportability functions $\sigma_{ij}$: $\kappa_j(v \in V_j) \in \sigma_{ij}(u \in V_i)\ \forall (u, v) \in A$. The Decompose-cluster algorithm selects each cluster (node set $W^x_i$ $(1 \leq i \leq k,\ 1 \leq x \leq \gamma_i)$) and collects MAs $(u, v)$ whose sources are in the cluster ($u \in W_i^x$), or whose destinations and sources are respectively in the cluster $W^x_i$ and in a cluster in $\sigma_{ij}(v)$ ($v \in W_i^x$ \& $\kappa_j(u \in V_j) \in \sigma_{ij}(v))$). Using the collected MAs and their endpoints, a subgraph $Y$ for reallocations of MAs in $W^x_i$ is composed. Each subgraph for each cluster is given to the $\Delta H$-algorithm so that it can improve the survivability by restructuring dependencies in the subgraph.

As will be understood, no directed cycles exist if no MA matches the condition of $v \in W_i^x$ \& $\kappa_j(u \in V_j) \in \sigma_{ij}(v)$. However, this is not going to happen in our work due to the assumption mentioned in Section \ref{assumptions}. Note that the absence of such MAs means that nodes in a cluster $x$ are not provided any support by the nodes that receive some supports from the nodes in the cluster $x$.

\begin{algorithm}[t]
\begin{algorithmic}[1]
\REQUIRE interdependent network (directed graph) $G = (V = \bigcup_{i=1}^k V_i, A)$, clustering functions $\kappa_i$
\STATE $D \gets \emptyset$
\FOR{a node set $W^x_i$ $(1 \leq i \leq k,\ 1 \leq x \leq \gamma_i)$} 
	\STATE $P \gets \emptyset,\ R \gets \emptyset$
	\FOR{each $(u, v) \in A \setminus D$}
		\IF{$u \in W_i^x$ or ($v \in W_i^x$ \& $\kappa_j(u \in V_j) \in \sigma_{ij}(v))$} 
			\STATE $P \gets P \cup \{u, v\}$		
			\STATE $R \gets R \cup (u, v)$
			\STATE $D \gets D \cup (u, v)$			
		\ENDIF
	\ENDFOR
	\STATE compose graph $Y = (P, R)$
	\STATE $\Delta H$-algorithm($Y, l$)
\ENDFOR
\end{algorithmic}
\caption{Decompose-cluster($G$)}
\label{clusterdecompose}
\end{algorithm}

\subsection{Complexity Analysis}
The Decompose-cluster algorithm extracts $\sum_{i=1}^k \gamma_i$ subgraphs from a given graph $G = (V, A)$. The number of clusters $\gamma_i$ in each constituent graph tends to be much smaller than the number of nodes; thus, $\sum_{i=1}^k \gamma_i$ can be considered as a constant. In order to compose each subgraph, the algorithm requires to check the source and destination of each arc in $A$. However, each edge appears in exactly one subgraph because of the used edge set $D$. Therefore, the total complexity of the Decompose-cluster algorithm is $O(|V| + |A|)$.

The complexity of the $\Delta H$-algorithm is sensitive to the number of cycles in the interdependent network. It is known that Johnson's algorithm finds all elementary cycles within $O((|V| + |E|)(|\mathcal{C}(G)| + 1))$. The $\Delta H$-algorithm determines a new destination after $\frac{l}{2} \times \mathcal{C}(G)$ searches for each MA, in the worst case. When only one cycle whose size is 2 exists in the input, and the other nodes are supported by the cycle, the size of the set $M$ becomes $|E| - 2$. It is obvious that the complexity of the Minimal-add process is $O(1)$, so the worst case analysis takes the case where all MAs are reallocated by the $\Delta H$-algorithm. Thus, its complexity is $O((|V| + |E|)(|\mathcal{C}(G)| + 1)) + O((|E| - 2) (\lceil{\frac{l}{2}} \rceil \times |\mathcal{C}(G)|))$. Assuming the maximum hop $l$ is small enough to be considered as a constant, the overall complexity of our heuristic algorithm becomes $O((|V| + |E|) |\mathcal{C}(G)|)$. Note that the assumption on $l$ is valid with our strategy, which tries to increase disjoint directed cycles in a given graph.

\subsection{Optimality in Special Graphs}
In order to analyze the performance of our heuristic algorithm, we consider the survivability improvement in special graphs where either an exhaustive search gives us the optimum survivability, or some special properties allow us to compute the optimum. 

In the analysis, the upper bound of the survivability improvement, which is used as a benchmark for the rest of this paper, is calculated based on the number of the MAs that satisfy the following two conditions. First, let $V_s$ be a set of nodes that hold more than one MA, and $M_s$ be a set of MAs whose source nodes are in $V_s$. Even when the MAs from $v \in V_s$ form more than one new cycles, the removal of such a source node $v$ can destroy all the newly formed cycles. This indicates that restructuring increases the survivability by at most $|V_s|$, when relocating MAs in $M_s$. Second, let $V_d$ be a set of nodes whose incoming arcs are all MAs, and $M_d$ be a set of MAs whose destination nodes are in $V_d$. If all the MAs incident to $v \in V_d$ are relocated, $v$ loses its functionality during this restructuring. Therefore, at least one MA should remain as an incoming arc to $v$. This implies that the number of cycles newly formed by the MAs in $M_d$ is at most $|M_d| - |V_d|$. Thus, the upper bound $U$ is obtained by $|M| - |M_s| + |V_s| - |V_d|$.

Fig. \ref{optresult} illustrates a comparison of our algorithm with the optimum solution in a small interdependent network such that each constituent graph has 15 nodes, and the number of dependency arcs is 84, including 5 MAs. The optimum solution is obtained by an exhaustive search of 759,375 combinations of reallocations. This numerical example shows that the solution given by the $\Delta H$ algorithm would not provide solutions that are exceptionally divergent from the optimum solution. It also infers that the upper bound is not tight in general.

\begin{figure}[t]
\centering
\includegraphics[width=.31\textwidth]{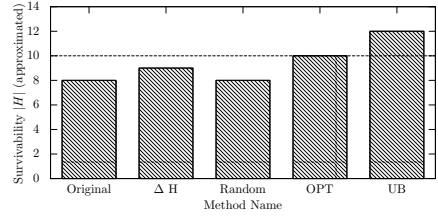}
\caption{Numerical comparison with the optimum solution in a small interdependent network.\label{optresult}}
\end{figure}
\begin{figure}[t]
\centering
\includegraphics[width=.31\textwidth]{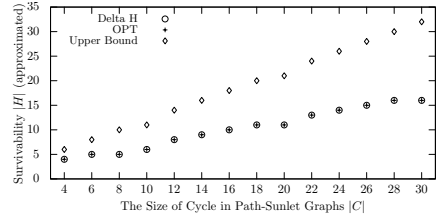}
\caption{Survivability of MA-saturated Path-Sunlet graphs $\zeta_2(G \in \mathcal{L})$ with two length-3 paths: $|\mathcal{P}| = 2$, $k_i = 3\ (\forall P_i \in \mathcal{P})$. \label{pathsunlet}}
\end{figure}

Fig. \ref{pathsunlet} indicates that the survivability obtained by our restructuring heuristic algorithm matches the optimum in a special class of graphs, which are named MA-saturated Path-Sunlet Graphs $\zeta_2(G),\ G \in \mathcal{S}$. The optimum value of survivability for these graphs is always computable based on the following discussion.

\begin{definition} \textcolor{black}{Path-Sunlet Graphs $\mathcal{L}$: 
A set of graphs satisfying the following conditions are named Path-Sunlet graphs. Let $\mathcal{L}$ denote the set of Path-Sunlet graphs.}
\textcolor{black}{\begin{itemize}
\item $G \in \mathcal{L}$ only has one cycle $C$.
\item The arcs that are not in the cycle $C$ form a set of disjoint paths whose initial nodes are in $C$: $\mathcal{P} = \{P_i = (v^i_1, v^i_2, ..., v^i_{k_i}) \mid v^i_1 \in C \textrm{ and } P_i \cap P_j = \emptyset\ (\forall P_{j \neq i} \in \mathcal{P}) \}$.
\end{itemize}}
\end{definition}

\begin{definition} \textcolor{black}{MA-saturation $\zeta_\delta(G)$ of a graph $G$:
The MA-saturation is an operation of adding additional arcs to a given graph until any addition of an arc makes the graph non-simple, maintaining the out-degree constraint that the out-degree of any node does not exceed a given constant $\delta \in \mathbb{N}$.}
\end{definition}

\begin{remark*}
\textcolor{black}{The optimal restructuring of MAs in MA-saturated Path-Sunlet graphs $\zeta_2(G),\ G \in \mathcal{L}$ consists of forming length-2 cycles using an MA and an edge in either $P_i \in \mathcal{P}$ or $C$.}
\end{remark*} 

\textcolor{black}{We consider the cases where $|\mathcal{P}| \geq 1$, because the survivability in the case of $|\mathcal{P}| = 0$ is obviously $\left\lceil \frac{|V(C)|}{2} \right\rceil$.}

\begin{lemma}\textcolor{black}{
By removing arcs that are not in any cycle, the optimally restructured MA-saturated Path-Sunlet graph $\zeta_2(G)$ is decomposed into some sequence of cycles.}
\end{lemma}

\begin{proof}\textcolor{black}{
Three or more cycles do not meet at the same node, since $\delta = 2$. Therefore, the only possible topology with multiple length-2 cycles is a chain of cycles, in which two cycles share exactly one node.}
\end{proof}

\begin{lemma}\textcolor{black}{
The survivability of the optimally restructured MA-saturated Path-Sunlet graphs $\zeta_2(G),\ G \in \mathcal{L}$ is $\sum_{q \in Q} \left\lceil \frac{q}{2} \right\rceil$, where $Q$ is the set of all the sequences of cycles obtained by removing the arcs that are not in any cycles.}
\end{lemma}

\begin{proof}\textcolor{black}{
A removal of one node that is shared by two cycles breaks the two cycles. When $q$ is even, the process gives us the survivability of $\frac{q}{2}$. If $q$ is odd, one additional removal is needed to destroy the remaining cycle. Thus, the survivability of a sequence of $q$ cycles is $\left\lceil \frac{q}{2} \right\rceil$.}

\textcolor{black}{Since each sequence in $Q$ is disjoint with the other, the survivability of the entire graph is obtained by summing up the survivability of each sequence.}
\end{proof}

\section{Simulation}
In order to understand the performance of the proposed algorithm, our simulations are conducted in both non-clustered and clustered interdependent network models of different sizes. The results from the simplest cases where each constituent network only consists of one cluster (non-clustered) are first described, and the clustered cases follow.

\subsection{Network Topology}
The performance of the proposed algorithm is analyzed in random directed bipartite graphs that contain at least one directed cycle. Assuming the situation in which a current interdependent network is working normally, each node is either a member of some cycle or reachable from a node in a cycle through some directed path in the input graph. Because our algorithm only concerns the dependency arcs between 2 constituent graphs ($k = 2$), any interdependent network is represented as a directed bipartite graph whose arcs connect a pair of different types of nodes.

Each random bipartite graph is generated by specifying the following parameter: $V_i$, $\max_{v \in V} \indeg(v)$ and $\min_{v \in V} \indeg(v)$. In order to observe the performance in different conditions, experiments are conducted in symmetric and asymmetric interdependent networks. A symmetric interdependent network has constituent networks which each have identical number of nodes: $|V_1| = |V_2|$, while constituent networks of an asymmetric interdependent network have different number of nodes: $|V_1| = \frac{|V_2|}{q}\ (q \in \mathbb{N})$. The degree of each node is determined based on the uniform distribution between the given maximum and minimum incoming degree.

\subsection{Clustering Settings}
As the non-clustered cases have symmetric and asymmetric constituent graphs, clustered interdependent networks are also examined in three patterns of topology configurations. In our simulations, each constituent graph has three clusters: $W_i^1$, $W_i^2$ and $W_i^3$ $(i = 1, 2)$ (See Fig. \ref{cluster_example}). In symmetric cases, a pair of corresponding clusters in different constituent graphs have the same number of nodes: $W_1^x = W_2^x$,  while a cluster is half-sized to the corresponding cluster in the other constituent graph in asymmetric models: $W_1^x = \frac{W_2^x}{2}$.  

Also, Fig. \ref{cluster_example} illustrates the three models that have different dependency relationships indicated as arrows. Note that when an arrow is drawn from $W_i^x$ to $W_j^{x'}$, it means that the nodes in cluster $W_j^{x'}$ can have supports from the nodes in $W_i^x$. Model 1 consists only of the solid arrows, which means that each pair of corresponding clusters has dependency relationships. Model 2 has the dependencies illustrated by the solid and dashed arrows, while Model 3 has all the arrows (solid, dashed and dotted). A major difference between these models is the possibility for a network to have some directed cycles over three or more clusters. In Model 1 and 2, directed cycles are able to exist only in a subgraph consisting of $W_1^1$ and $W_2^1$, $W_1^2$ and $W_2^2$, or $W_1^3$ and $W_2^3$, while a directed cycle can lie over the entire graph containing all the clusters in Model 3.

\def\scale{0.9}
\begin{figure}[t]
\centering
\includegraphics[width=0.8\columnwidth]{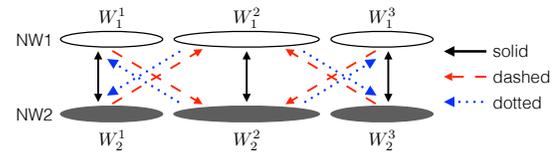}
\caption{Dependency models of clustered interdependent networks. Arrows show the dependency relationships between clusters. Model 1: solid. Model 2: solid and dashed. Model 3: solid, dashed, and dotted.}
\label{cluster_example}
\end{figure}

\begin{figure*}[t]
\centering
\begin{minipage}[t]{0.31\textwidth}
\centering
\includegraphics[width=1\columnwidth]{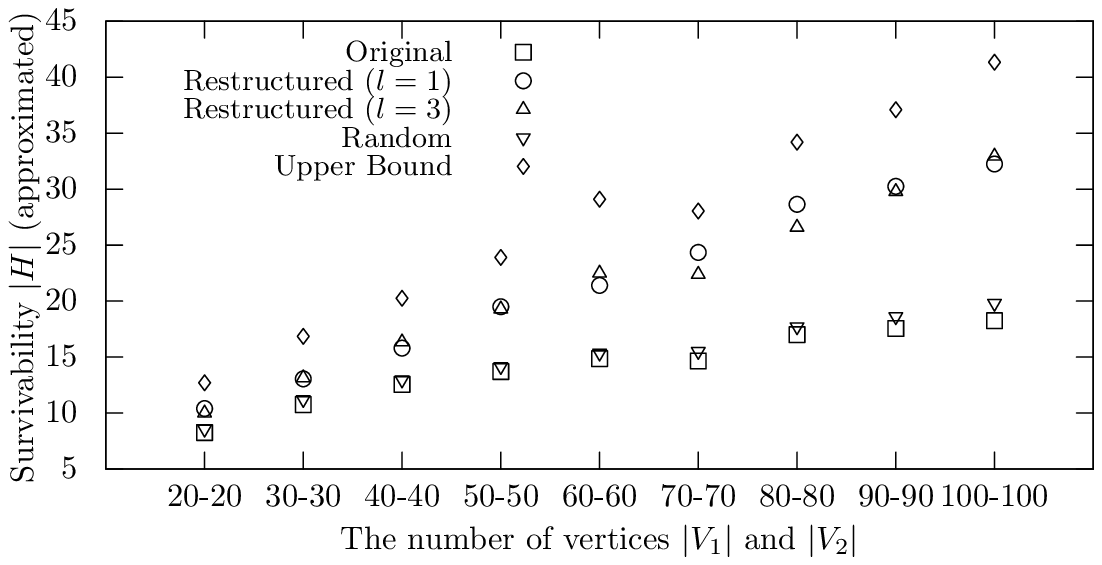}
\caption{Survivability of interdependent networks before and after the improvement under $|V_1| = |V_2|$, $\max_{v \in V} \indeg(v) = 4$, and $\min_{v \in V} \indeg(v) = 2$, and $l = 1, 3$.
\label{symmetric_result}}
\end{minipage}\hfill
\begin{minipage}[t]{0.31\textwidth}
\centering
\includegraphics[width=1\columnwidth]{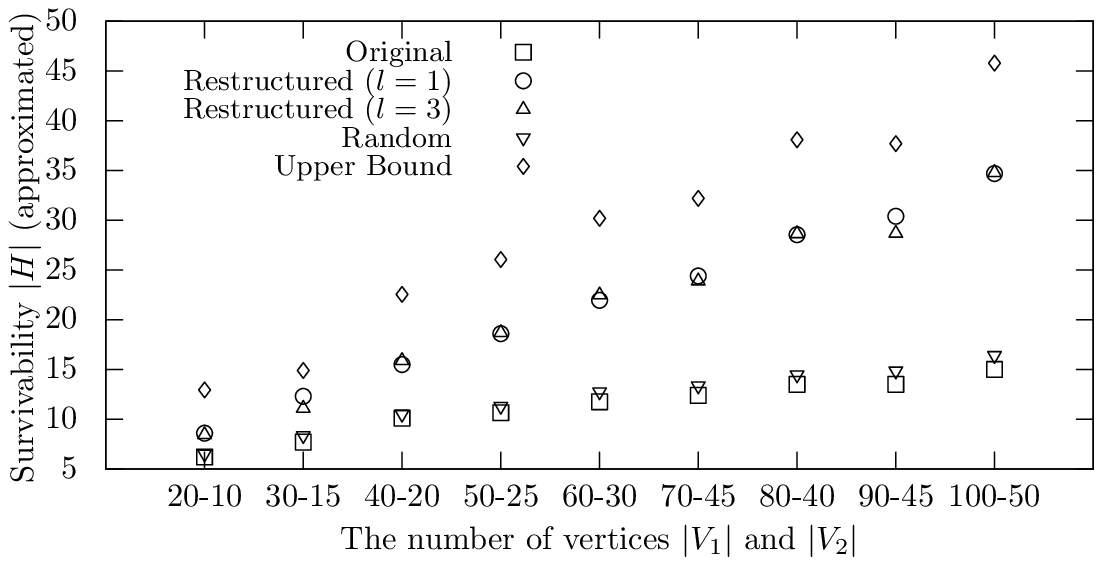}
\caption{Survivability of interdependent networks before and after the improvement under $|V_1| = \frac{|V_2|}{2}$, $\max_{v \in V} \indeg(v) = 4$, $\min_{v \in V} \indeg(v) = 2$, and $l = 1, 3$.
\label{asymmetric_result}}
\end{minipage}\hfill
\begin{minipage}[t]{0.31\textwidth}
\centering
\includegraphics[width=1\columnwidth]{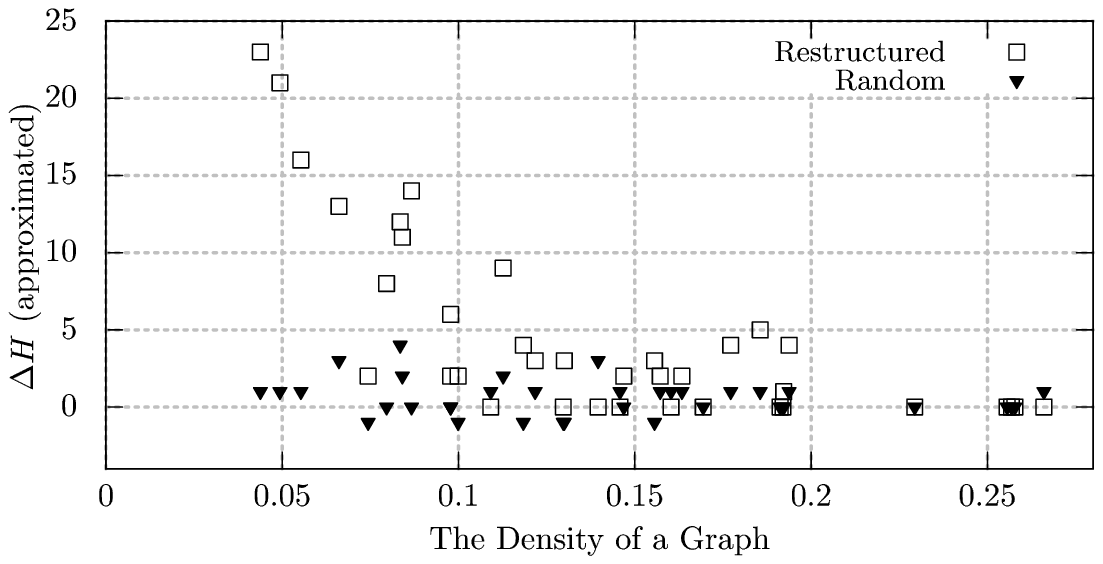}
\caption{The relationship between graph density and $\Delta H$.\label{density_result}}
\end{minipage}
\end{figure*}
\begin{figure*}[t]
\centering
\begin{minipage}[t]{0.31\textwidth}
\centering
\includegraphics[width=1\columnwidth]{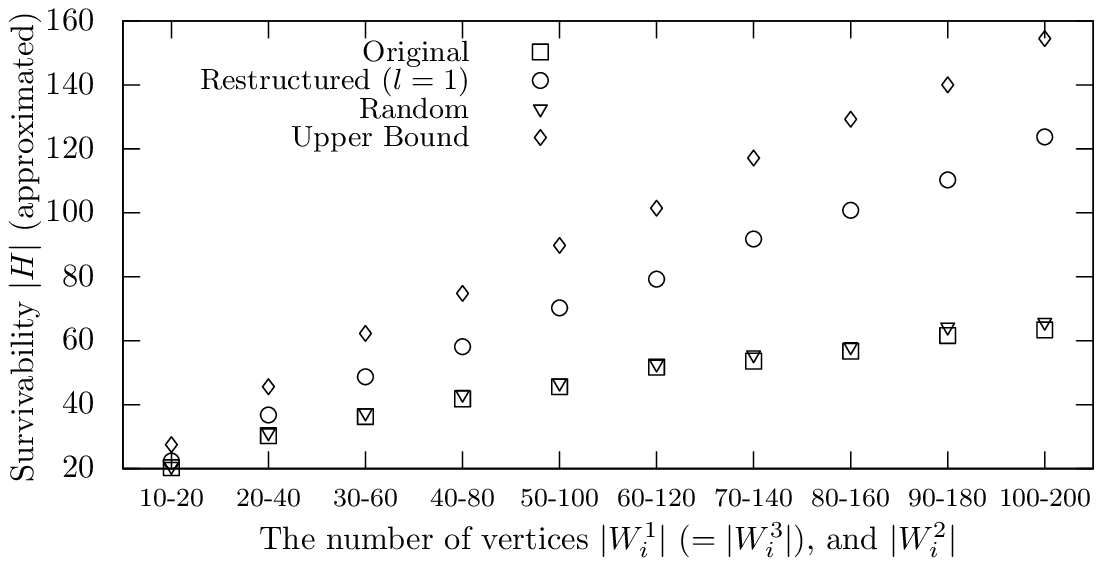}
\caption{Survivability of clustered interdependent networks (Model 2) before/after the improvement under $|W_1^1|  = |W_1^3| = |W_2^1| = |W_2^3|$, $|W_1^2| = |W_2^2|$, $\max_{v \in V} \indeg(v) = 4$, $\min_{v \in V} \indeg(v) = 2$, and $l = 1$.
\label{sym_clustered_result}}
\end{minipage}\hfill
\begin{minipage}[t]{0.31\textwidth}
\centering
\includegraphics[width=1\columnwidth]{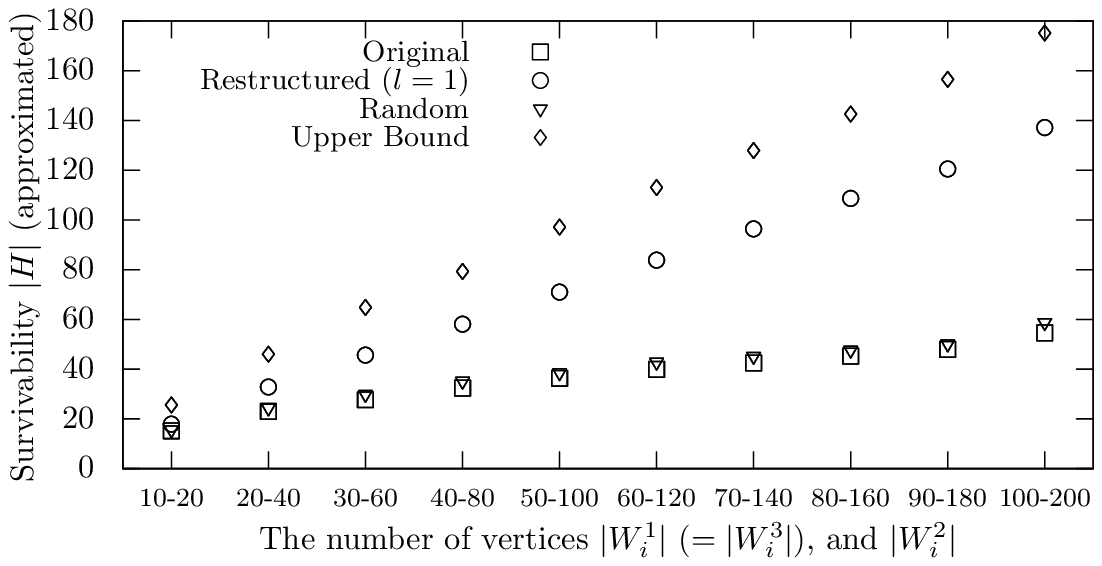}
\caption{Survivability of clustered interdependent networks (Model 2) before/after the improvement under $|W_1^1| = |W_1^3| = \frac{|W_2^1|}{2} = \frac{|W_2^3|}{2}$, $|W_1^2| = \frac{|W_2^2|}{2}$, $\max_{v \in V} \indeg(v) = 4$, $\min_{v \in V} \indeg(v) = 2$, and $l = 1$. \label{asym_clustered_result}}
\end{minipage}\hfill
\begin{minipage}[t]{0.31\textwidth}
\centering
\includegraphics[width=1\columnwidth]{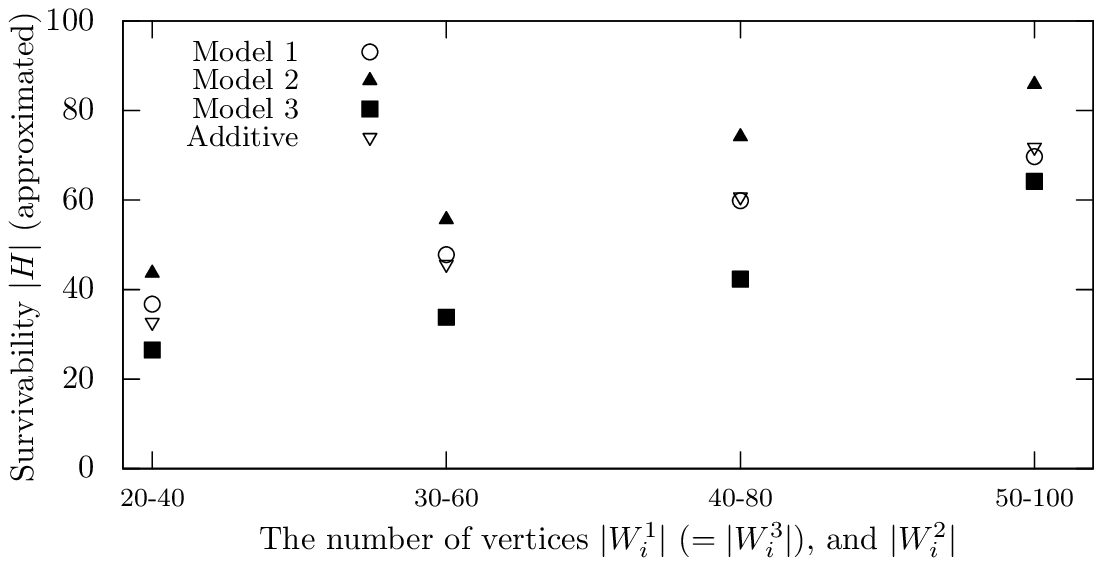}
\caption{Comparison of survivability among different dependency models.\label{additive}}
\end{minipage}
\end{figure*}

\subsection{Metrics}
The survivability of the given graphs, restructured graphs, randomly reassigned graphs, and the upper bound of the improvement are illustrated in our results. The random reassignments of MAs are conducted with a uniform distribution over all the nodes in the other constituent graph from the constituent graph that includes the source of an MA.

Computing the size of the cycle hitting set is known to be NP-complete even in bipartite graphs, so the exact value cannot be obtained in larger graphs. Our evaluation is conducted using a well-known approximation algorithm whose approximation factor is $\ln{|V|} + 1$ \cite{greedy_setcover}.

Furthermore, the density of a given graph $G = (V, A)$ defined by $\frac{|A|}{\prod_i |V_i|}$ is used to examine the relationship between the survivability improvement, and the maximum and minimum degrees.

\subsection{Results}
\subsubsection{Non-clustered Cases}
Figs. \ref{symmetric_result} and \ref{asymmetric_result} illustrate the survivability of the given and restructured graphs with identical and halved size constituent graphs, respectively. In both cases, our method demonstrates more improvement of the survivability compared to the random reassignment. The survivability of the original graphs $|H(G)|$ maintains a similar value regardless of the size of graphs, though the survivability of the graphs restructured by our method $|H(G')|$ steeply increases along with the size of the graph. Since, in the original graph $G$, arcs are randomly added, it could be difficult to form larger directed cycles. Therefore, it is reasonable that the number of disjoint cycles indicates the tendency to stay within a similar range of values. On the other hand, there would exist more MAs in larger graphs, because these graphs have more arcs that are not in directed cycles. This results in dramatic enhancement of the survivability in larger graphs. The difference caused by the given maximum hop $l$ for our algorithm remains small over all sizes of a graph.

Fig. \ref{density_result} indicates the relationship between the density of graphs and $\Delta H$, the amount of survivability improvement. We compare our method to the random reassignment. \textcolor{black}{The result shows that, in graphs with lower density, our method has greater success in increasing the survivability.} An observed general trend of our method is the gradual decrease in $\Delta H$ in accordance with the density. This trend seems to be induced by the fact that the graphs with more arcs have a higher possibility of composing cycles even in the original topology. This implies that graphs with higher density have fewer MAs that can form new disjoint cycles. On the other hand, the random reassignment does not demonstrate its effectiveness for the improvement in graphs with any density, which is the same result from Figs. \ref{symmetric_result} and \ref{asymmetric_result}. Moreover, the random reassignment sometimes decreases the survivability ($\Delta H < 0$). It is conceivable that the reassignment connects two (or more) cycles and make it possible to decompose all these cycles by the removal of a node. This result implies that imprudent restructuring of the dependency may cause more fragility of the interdependent networks.

\subsubsection{Clustered Cases}
The results in clustered interdependent networks whose dependency relationships follow Model 2 are shown in Figs. \ref{sym_clustered_result} and \ref{asym_clustered_result}. Similar trends to non-clustered cases are observed for both symmetric and asymmetric cases. The proposed method succeeds in increasing the survivability for different sizes of interdependent networks. 

Fig. \ref{additive} illustrates the difference in survivability after restructuring among the three types of dependency models of symmetric networks. The value of ``Additive" is obtained by the simple addition of non-clustered cases that jointly compose a clustered case. For instance, the case of clustered networks consisting of 20, 40, and 20 nodes clusters is compared with the sum of the survivability of the cases of non-clustered networks of 20, 40, and 20 nodes shown in Fig. \ref{symmetric_result}. The dependency relations among clusters increase from Model 1 to Model 3 (See Fig \ref{cluster_example}). 

Model 1 gives similar survivability to the simple addition of non-clustered cases, since a pair of corresponding clusters in two constituent graphs is independent from the other pairs in this model. In Model 2, the survivability of the entire network increases, because the nodes in cluster $W_i^2$ can have more supports from the clusters whose cycles are disjoint from the cycles in $W_i^2$. Although more supports exist among the clusters in Model 3, its survivability is less than the other models. In Model 3, a cycle can lie on more clusters because of the bidirectional dependencies among all the clusters. This topological characteristic is likely to increase the overlapping of multiple cycles and results in the decline of survivability in this model. These results cast a doubt on a naive statement claiming that the increase of dependencies induces more fragility in general interdependent networks.

\begin{figure}[t]
\centering
\includegraphics[width=0.38\textwidth]{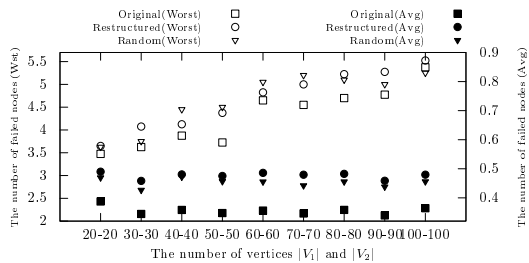}
\caption{The number of failed nodes (Worst case and Average) after a single node failure under $|V_1| = |V_2|$, $\max_{v \in V} \indeg(v) = 4$, $\min_{v \in V} \indeg(v) = 2$, and $l = 1$.\label{impact_result}}
\end{figure}

\section{Discussion: Impact Alleviation vs Survivability}
Although it is not the primary focus of this paper, in this section, we evaluate the behavior of the proposed algorithms in terms of its effect on the size or impact of a cascading failure. Fig. \ref{impact_result} illustrates the influence of our dependency modifications on the size of cascading failures induced by a single node. In this experiment, \textcolor{black}{the impact of a single node failure at a node $v$ is defined as the number of nodes $\theta_v$ that become nonfunctional after a cascading failure initiated by the failure of $v$.} The results are analyzed in terms of the following two metrics:
\begin{itemize}
\item \textit{Worst} (non-filled points): the size of the largest cascading failure: $\max_{v \in V} \theta_v$,
\item \textit{Average} (filled points): the average size of all possible cascading failures: $\frac{\sum_{v \in V} \theta_v}{|V|}$.
\end{itemize}

The robustness of restructured networks against a single node failure always declines in comparison with the original topology. The decline in the size of the largest cascading failure is most remarkable in the case of $|V_1| = |V_2| = 50$ in our simulation. In this case, the size of a cascading failure increases by 1 node after the restructuring. 

In general, the concentrations of provisioning on a certain portion of a network can improve the survivability, though it can make the other portions more fragile. In contrast, appropriate distributions of provisioning are necessary in order to alleviate the impact of any possible single node failure. This difference in robustness against single node failures and system survivability could be a reason for the decline. 

However, when examining the average size of cascading failures, it is observed that the increase in the average number of failed nodes is suppressed within $0.1$ nodes over all network sizes. Thus, it could be said that our method does not deteriorate the robustness against single node failures.

\section{Conclusion}
This paper addresses the design problem of survivable clustered interdependent networks under some constraints relating to the existence of legacy systems during restructuring. Based on the definition of the survivability proposed in a related work, it is claimed that the increase of disjoint cycles could enhance the survivability. The proposed heuristic algorithm tries to compose new disjoint cycles by gradual relocations of certain dependencies (Marginal Arcs) in order to guarantee the functionality of existing systems. Our simulations indicate that the algorithm succeeds in increasing the survivability, especially in networks with fewer dependencies. Moreover, the empirical result implies that the number of dependencies, in general, is not the root cause of the vulnerability to cascading failures. Rather, the appropriate additions of dependencies can improve the overall survivability, while poorly designed dependencies make networks more fragile. When redesigning the interdependency between control and functional entities in SDN, NFV, or CPSs based on the proposed algorithm, the possibility to experience catastrophic cascading failures would decrease.

\bibliographystyle{ieeetr} 
\bibliography{ref}

\begin{IEEEbiography}[{\includegraphics[width=1in,height=1.25in,clip,keepaspectratio]{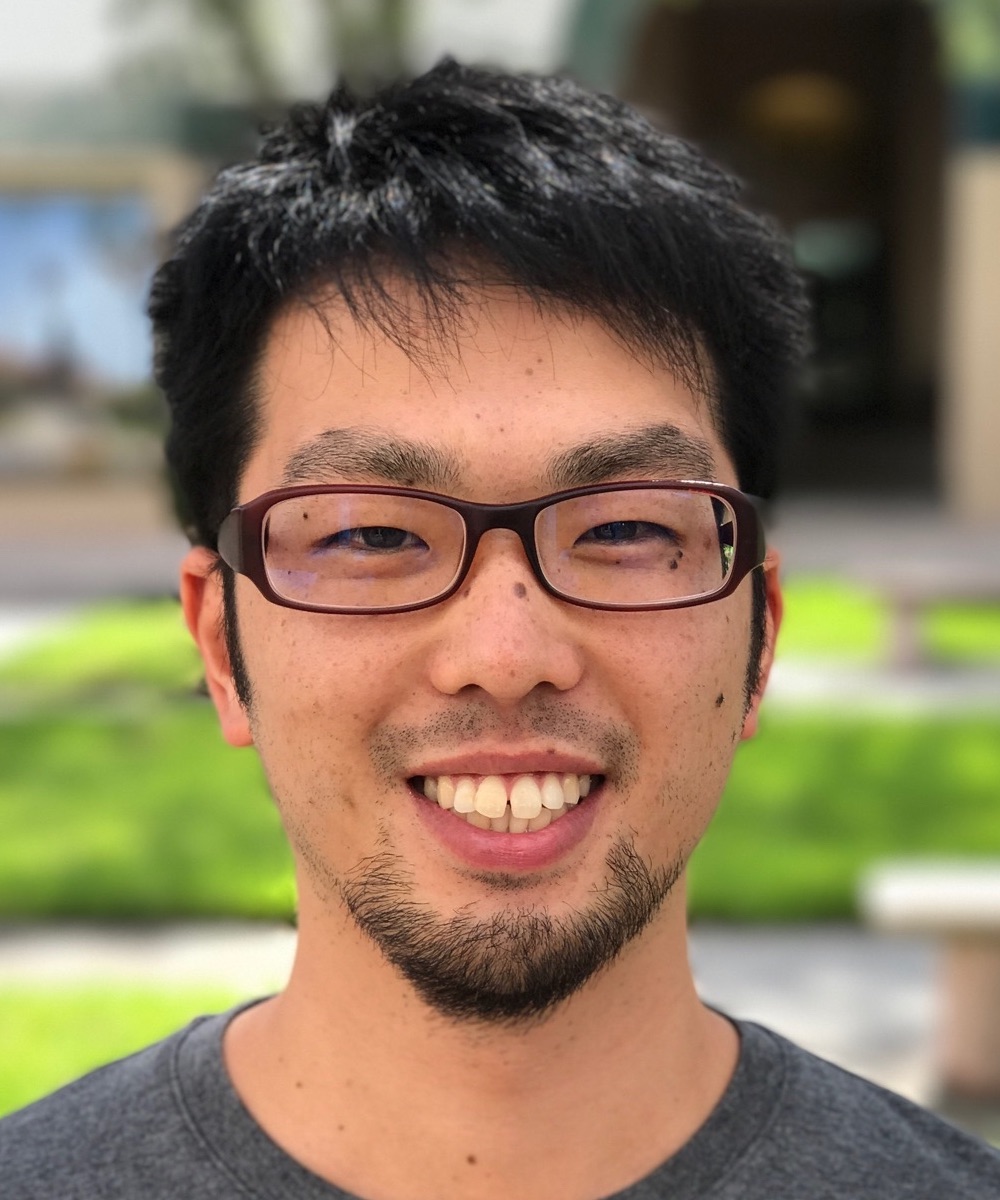}}]{Genya~Ishigaki} (GS'14) received the B.S. and M.S. degrees in engineering from Soka University, Tokyo, Japan, in 2014 and 2016, respectively. He is currently pursuing the Ph.D. degree in computer science at Advanced Networks Research Laboratory, The University of Texas at Dallas, Richardson, TX, USA. His current research interests include design and recovery problems of interdependent networks, and software defined networking.
\end{IEEEbiography}

\begin{IEEEbiography}[{\includegraphics[width=1in,height=1.25in,clip,keepaspectratio]{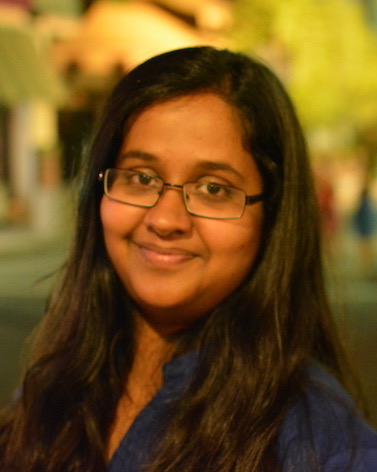}}]{Riti~Gour}
received her BE degree in Electronics and Telecommunication Engineering from S.S.C.E.T, Bhilai, India, in 2012, and her MS degree in Telecommunications Engineering from the University of Texas at Dallas, Texas, in 2015. Since 2015, she has been working towards her Ph.D. degree at UT Dallas, majoring in telecommunications. Her research is focused towards survivability of optical networks against correlated failures and disasters using graph optimization and machine learning techniques. 
\end{IEEEbiography}

\begin{IEEEbiography}[{\includegraphics[width=1in,height=1.25in,clip,keepaspectratio]{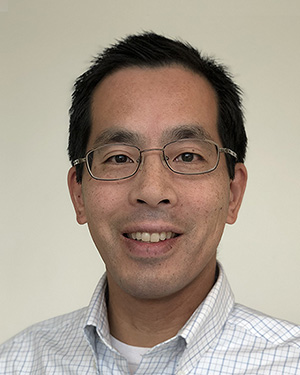}}]{Jason~P.~Jue} (M'99-SM'04) received the B.S. degree in Electrical Engineering and Computer Science from the University of California, Berkeley in 1990, the M.S. degree in Electrical Engineering from the University of California, Los Angeles in 1991, and the Ph.D. degree in Computer Engineering from the University of California, Davis in 1999. He is currently a Professor in the Department of Computer Science at the University of Texas at Dallas. His current research interests include optical networks and network survivability.
\end{IEEEbiography}

\end{document}